\documentclass[twoside,leqno]{article}
\usepackage[numbers,sort]{natbib}
\pdfoutput=1
\usepackage{mathtools}
\usepackage[mathcal]{euscript}
\usepackage{nicefrac}
\usepackage{cs_ops}
\usepackage{url,color}
\usepackage{morefloats}
\usepackage{enumitem}
\usepackage{color}
\usepackage{balance}
\usepackage{ifpdf}

\newboolean{fullversion}
\setboolean{fullversion}{true}

\ifthenelse{\boolean{fullversion}}
{
\usepackage{amsthm}
\usepackage[vmargin=1.0in,hmargin=1.0in]{geometry}
\newtheorem{lemma}{Lemma}[section]
}
{
\usepackage{ltexpprt}
\twocolumn
}

\newcommand{\ignore}[1]{}
\newcommand{\comment}[1]{\textcolor{red}{[#1]}}

\DeclareMathOperator{\InsertOp}{\textsc{Insert}}
\newcommand{\Insert}[3]{\InsertOp\left({#1,#2,#3}\right)}
\DeclareMathOperator{\DecreaseKeyOp}{\textsc{DecreaseKey}}
\newcommand{\DecreaseKey}[3]{\DecreaseKeyOp\left({#1,#2,#3}\right)}
\DeclareMathOperator{\DeleteMinOp}{\textsc{DeleteMin}}
\newcommand{\DeleteMin}[1]{\DeleteMinOp\left({#1}\right)}
\newcommand{\rank}[1]{r({#1})}
\newcommand{\maxrank}[1]{m({#1})}

\begin{document}

\title{A Back-to-Basics Empirical Study of Priority Queues\thanks{Research at
Princeton University partially supported by NSF	grant CCF-0832797 and a Google
PhD Fellowship.}}
\author{Daniel~H.~Larkin\thanks{Princeton University,
Department of Computer Science.
	Email: {\tt dhlarkin@cs.princeton.edu}.}
	\and
	Siddhartha Sen\thanks{Microsoft Research SVC. Email:
	\texttt{sisen@microsoft.com}}
	\and
	Robert~E.~Tarjan\thanks{Princeton University, Department of Computer Science
	and Microsoft Research SVC. Email: {\tt ret@cs.princeton.edu}.}}

\maketitle

\begin{abstract}
The theory community has proposed several new heap variants in the recent past
which have remained largely untested experimentally.  We take the field back to
the drawing board, with straightforward implementations of both classic and
novel structures using only standard, well-known optimizations.  We study the
behavior of each structure on a variety of inputs, including artificial
workloads, workloads generated by running algorithms on real map data, and workloads from a
discrete event simulator used in recent systems networking research.  We provide
observations about which characteristics are most correlated to performance.
For example, we find that the L1 cache miss rate appears to be strongly
correlated with wallclock time.  We also provide observations about how the input sequence
affects the relative performance of the different heap variants.  For example,
we show (both theoretically and in practice) that certain random
insertion-deletion sequences are degenerate and can lead to misleading results.
Overall, our findings suggest that while the conventional wisdom holds in some
cases, it is sorely mistaken in others.
\end{abstract}

\section{Introduction}

The priority queue is a widely used abstract data structure.  Many theoretical
variants and implementations exist, supporting a varied set of operations with differing
guarantees.  We restrict our attention to the following base set of commonly
used operations:

\begin{itemize}[leftmargin=*]
\setlength{\itemsep}{0pt}
\setlength{\parskip}{0pt}
	\item $\Insert{Q}{x}{k}$ --- insert item $x$ with key $k$ into heap $Q$ and
	return a handle $\bar x$
	\item $\DeleteMin{Q}$ --- remove the item of minimum key from heap $Q$ and
	return its corresponding key $k$
	\item $\DecreaseKey{Q}{\bar x}{k'}$ --- given a handle $\bar x$, change the key
	of item $x$ belonging to heap $Q$ to be $k'$, where $k'$ is guaranteed to be less than the original key $k$
\end{itemize}

It has long been known that either $\InsertOp$ or $\DeleteMinOp$ must take
$\OW{\log n}$ time due to the classic lower bound for sorting
\cite{KnuthSorting}, but that the other operations can be done in $\OO{1}$ time.
In practice, the worst-case of $\log n$ is often not encountered or can be
treated as a constant, and for this reason simpler structures with logarithmic
bounds have traditionally been favored over more complicated,
constant-time alternatives.
In light of recent developments in the theory community
\cite{ElmasryViolation,EdelkampWeak,HaeuplerRankPairing,ChanQuake,BrodalStrict}
and the outdated nature of the most widely cited experimental studies on
priority queues \cite{MoretExp,StaskoExp,LadnerCacheHeaps}, we aim to revisit
this area and reevaluate the state of the art.  More recent
studies~\cite{BruunPolicy,ElmasryFat,EdelkampWeak} have been narrow in focus
with respect to the implementations considered (e.g., comparing a single new heap to a few classical ones), the workloads tested (e.g.,
using a few synthetic tests), or the metrics collected (e.g.,
measuring wallclock time and element comparisons).  In addition to the normal
metric of wallclock time, we have collected additional metrics such as branching
and caching statistics.  Our goal is to identify experimentally verified trends
which can provide guidance to future experimentalists and theorists alike.  We
stress that this is not the final word on the subject, but merely another line
in the continuing dialogue.

In implementing the various heap structures, we take a different approach
from the existing algorithm engineering literature, in that we do not perform
any algorithm engineering. That is, our implementations are intentionally
straightforward from their respective descriptions in the original papers.
The lack of considerable tweaking and algorithm engineering in this
study is, we believe, an example of na\"ivet\'e as a virtue.  We expect that
this would accurately reflect the strategy of a practitioner seeking to make
initial comparisons between different heap variants.  As a sanity check, we also
compare our implementations with a state-of-the-art, well-engineered
implementation often cited in the literature.
%, and find that our fastest heaps are competitive with it.

Our high-level findings can be summarized as follows.  We find that wallclock
time is highly correlated with the cache miss rate, especially in the L1 cache.
High-level theoretical design decisions---such as whether to use an
array-based structure or a pointer-based one---have a significant impact on
caching, and which design fares best is dependent on the specific workload.
For example, Fibonacci heaps sometimes outperform implicit $d$-ary heaps, in
contradiction to conventional wisdom.  Even a well-engineered implementation
like Sanders' sequence heap~\cite{SandersSequence} can be bested by our
untuned implementations if the workload favors a different method.

Beyond caching behavior, those heaps with the simplest implementations tend to
perform very well.  It is not always the case that a theoretically superior
or simpler structure lends itself to simpler code in practice.
Pairing heaps dominate Fibonacci heaps across the board, but
interestingly, recent theoretical simplifications to Fibonacci heaps tend to do
worse than the original structure.

Furthermore we found that a widely-used benchmarking workload is degenerate in a
certain sense.  As the sequence of operations progresses, the distribution of
keys in the heap becomes very skewed towards large keys, contradicting the
premise that the heap contains a uniform distribution of keys. This can be shown
both theoretically and in practice.

Our complete results are detailed in Sections~\ref{sec:results}
and~\ref{sec:sanity}. We first describe the heap variants
we implemented in Section~\ref{sec:heaps}, and then discuss our experimental
methodology and the various workloads we tested in Section~\ref{sec:exp}. We
conclude in Section~\ref{sec:remarks} with some remarks.

\section{Heap Variants}
\label{sec:heaps}

Aiming to be broad, but not necessarily comprehensive, this study includes both
traditional heap variants and new variants which have not previously
undergone much experimental scrutiny.  We have implemented the following
structures, listed here in order of program length: implicit $d$-ary heaps, pairing heaps, Fibonacci heaps,
binomial queues, explicit $d$-ary heaps, rank-pairing heaps, quake heaps,
violation heaps, rank-relaxed weak queues, and strict Fibonacci heaps.
Table~\ref{tab:lloc} lists the logical lines of code; in our
experience, this order corresponded exactly to perceived programming difficulty.
There are several other heap variants which may be worth investigating, but
which have not been included in this study.  Among those not included are the 2-3 heap \cite{Takaoka23},
thin/thick heaps \cite{KaplanThin}, and the buffer heap \cite{ChowdhuryBuffer}.

Williams' binary heap \cite{WilliamsBinary} is the textbook example of a
priority queue.  Lauded for its simplicity and taught in undergraduate computer
science courses across the world, it is likely the most widely used variant
today.  Storing a complete binary tree whose nodes obey the heap order gives a
very rigid structure;
% which can be exploited for computational efficiency.
indeed, the heap supports all operations in worst-case $\OT{\log n}$ time.
The tree can be stored explicitly using heap-allocated nodes and pointers, or it
can be encoded implicitly as a level-order traversal in an array.  We refer to these variations as \textit{explicit} and \textit{implicit} heaps respectively.
The implicit heap carries a small caveat, such that in order to support
$\DecreaseKeyOp$ efficiently, we must rely on a level of indirection: encoding
the tree's structure as an array of node pointers and storing the current index
of a node's pointer in the node itself (allowing us to return the node pointer
to the client as $\bar x$).  This study includes two versions of implicit
heaps---one that supports $\DecreaseKeyOp$ through this indirection, and one
that doesn't.  We refer to the latter as the \textit{implicit-simple} heap.

Explicit and implicit heaps can be generalized beyond the binary
case to have any fixed branching factor $d$.  We refer to these heaps
collectively as $d$-ary heaps; this study examines the cases where
$d=2,4,8,16$.
To distinguish between versions with different branching factors, we label
the heaps in this fashion: \textit{implicit-2}, \textit{explicit-4},
\textit{implicit-simple-16}, and so forth.

\begin{table}[t]
\center
\caption{Programming effort}
\label{tab:lloc}
\vspace{5pt}
\begin{tabular}{l|c}
Heap variant & Logical lines of code (lloc) \\ \hline
implicit simple & 184 \\
pairing & 186\\
implicit & 194\\
Fibonacci & 282\\
binomial & 317 \\
explicit & 319 \\
rank-pairing & 376\\
quake & 383\\
violation & 481\\
rank-relaxed weak & 638\\
strict Fibonacci & 1009\\
\end{tabular}
\end{table}

Beyond the $d$-ary heaps, all other heap variants are primarily pointer-based
structures, though some make use of small auxiliary arrays.  All are conceptual
successors to Vuillemin's binomial queue \cite{VuilleminBinomial}.  Originally
developed to support efficient melding (which takes linear time in
$d$-ary heaps), we have included it in our study due to its simplicity.  The binomial
queue stores a forest of perfect, heap-ordered binomial trees of unique rank.
This uniqueness is maintained by linking trees of equal rank such that the root
with lesser key becomes the new parent.  To support deletion of a node, each of
its children is made into a new root, and the resulting forest is then
processed to restore the unique-rank invariant.  This can lead to a
fair amount of structural rearrangement, but the code to do so is
rather simple.  Key decreases are handled as in $d$-ary heaps by
sifting upwards.  Like $d$-ary heaps, binomial queues support all operations in
worst-case $\OT{\log n}$ time.

Most other heap variants can be viewed as some sort of relaxation of the
binomial queue, with the chronologically first one being the Fibonacci heap
\cite{FredmanFibonacci}.  The Fibonacci heap achieves amortized $\OO{1}$-time
$\InsertOp$ and $\DecreaseKeyOp$ by only linking after deletions and allowing
some imperfections in the binomial trees.  The imperfections are generated by
key decreases: instead of sifting, a node is cut from its parent and made into a new root.  To prevent the trees from becoming too malformed, a node is also cut from its parent as soon as it loses a second child.  This can lead to a
series of upwardly cascading cuts.  The violation heap \cite{ElmasryViolation}
and rank-pairing heaps \cite{HaeuplerRankPairing} can be viewed as further
relaxations of the Fibonacci heap.  The rank-pairing heaps allow rank
differences greater than one, and propagate ranks instead of cascading cuts so
that at most one cut is made per $\DecreaseKeyOp$.  Two rank rules were proposed
by the authors, leading to our implementations being labeled
\textit{rank-pairing-t1} and \textit{rank-pairing-t2}.  The violation heap also
propagates ranks instead, only considering rank differences in the two most
significant children of a node.  It allows two trees of each rank and utilizes a
three-way linking method.  The pairing heap \cite{FredmanPairing} is essentially
a self-adjusting, single-tree version of the Fibonacci heap, where ranks are not
stored explicitly, and linking is done eagerly.  Its amortized complexity is
still an open question, though it has been shown that $\DecreaseKeyOp$ requires
$\OW{\log \log n}$ time if all other operations are $\OO{\log n}$
\cite{FredmanBound}.  Two different amortization arguments can be used to prove
either $\OO{1}$ and $\OO{\log n}$ bounds for $\InsertOp$ and $\DecreaseKeyOp$
respectively \cite{IaconoBound} or $\OO{2^{2\sqrt{\log \log n}}}$ for both
operations \cite{PettieBound}.  It remains an open question to prove an
$\oo{\log n}$ bound for $\DecreaseKeyOp$ simultaneously with $\OO{1}$-time
$\InsertOp$.  All three relaxations are intended to be in some way simpler than
Fibonacci heaps, with the hope that this makes them faster in practice.

The strict Fibonacci heap \cite{BrodalStrict} on the other hand, intends to
match the Fibonacci time bounds in the worst case, rather than in an amortized
sense.  This leads to a fair amount of extra code to manage structural
imperfections somewhat lazily.  Rank-relaxed weak queues
\cite{EdelkampWeak} are essentially a tweaked version of rank-relaxed heaps,
with an emphasis on minimizing key comparisons.  They mark nodes as potentially
violating after a $\DecreaseKeyOp$ operation and clean them up lazily.  Quake
heaps \cite{ChanQuake} are a departure from the Fibonacci model, but are still
vaguely reminiscent of binomial queues.  A forest of uniquely-ranked tournament
trees is maintained.  Subtrees may be missing, but the number of nodes at a
given height decays exponentially in the height, a property guaranteed through a
set of global counters and a global rebuilding process triggered after
deletions.  There are multiple implementation strategies mentioned in the
original paper, but only the one that was fully detailed (the full
tournament representation) has been implemented here.  It is possible that the
other implementations would be more efficient.

\section{Experimental Design and Workloads}
\label{sec:exp}

Our codebase is written primarily in C99 and is available online for inspection,
modification, and further development~\cite{code}.
As we stated earlier, our implementations are intentionally straightforward from
their respective descriptions in the original papers, or use only the most
basic, well-known optimizations for the more studied structures.  Further optimization is left to
the compiler (\texttt{gcc -O4}) so as not to unfairly bias toward one variant or
another.

%The lack of considerable tweaking and algorithm engineering present in this
%study is, we believe, an example of na\"ivet\'e as a virtue.
%Optimization is left to the compiler (\texttt{gcc -O4}) so as not to unfairly
%bias toward one variant or another.  We expect that this would accurately
%reflect the strategy of a practitioner seeking to make initial comparisons
% between different heap variants.

Keys are 64-bit unsigned integers, while the items
themselves are 32-bit unsigned integers.  In most cases,
the key actually consists of a 32-bit key in the high-order bits and the item
identifier in the low-order bits, in order to break ties during comparisons.

We experimented with different memory allocation schemes using our own
simple fixed-size memory pool implementation.  This abstraction layer allowed us
to allocate all memory eagerly using a single \texttt{malloc}, allocate lazily by
doubling when space fills, or allocate completely on the fly using a
\texttt{malloc} for each $\InsertOp$.  In our experiments, the memory allocation
scheme made very little difference regardless of heap variant, indicating that
this layer of optimization was superfluous. Thus, all the results in this paper use the eager
strategy.

The workloads we tested are described in the subsections below. These include
workloads generated by code sourced (with modifications) from
DIMACS implementation challenges~\cite{DIMACS5,DIMACS9}, as well as workloads
generated by a packet-level network simulator~\cite{htsim}.
All experiments use trace-based simulation.  More specifically, a workload is
generated once using a reference heap and the sequence of operations and values
is recorded in a trace file.  This trace file can then be executed against each
of many drivers---one for each heap variant included in the study---as well as a
dummy driver that simply parses the trace file but does not execute any heap
operations.  The dummy driver captures the overhead of the simulation and its
collected metrics are subtracted from those of the other drivers before any
comparisons are done. Wallclock time is measured by the driver itself.  For
purposes of timing, each execution of a trace file is run for a minimum of five
iterations and two seconds of wallclock time (whichever takes longer), and the
time is averaged over all iterations.  Other metrics are collected over the
course of a single iteration using \texttt{cachegrind}~\cite{cachegrind}, a cache and branch-prediction profiler.
The profiler simulates actual machine parameters and does not vary between
executions, providing accurate measurements that are isolated from other system
processes.  We have used it to collect dynamic instruction and
branching counts as well as reads, writes and misses for both the L1 and L2
caches. Additionally, \texttt{cachegrind} allows for simulating branch
prediction in a basic model (that does not correspond exactly to the real
machines); we have collected this misprediction count as well.

All experiments were run on a high-performance computing cluster in Princeton
consisting of Dell PowerEdge SC1435 nodes with dual AMD Opteron
2212 processors (dual-core, 2.0GHz, 64KB L1 cache and 1MB L2 cache per core) and
8GB of RAM (DDR2-667). The machines ran Springdale/PUIAS Linux (a Red-Hat
Enterprise clone) with kernel version 2.6.32.  All executions remained in-core.

\subsection{Artificial randomized workloads.}
\label{sec:artificial}

The first workload we consider is a standard and ubiquitous one: sorting sequences of
$n$ uniformly random integers.  This translates to $n$ random insertions
followed by $n$ minimum deletions in the trace files.

The next type of sequence intermixes insertions and deletions, but in a very
structured way which turns out to be degenerate.  It is a very
natural sequence to test, and due to its presence in the DIMACS test set, we
worry that its use in benchmarks may be more widespread than one might hope of a
broken test.  The sequence begins with $n$ random insertions as in the sorting
case.  It is then followed by $cn$ repetitions of the following: one random
insertion followed by one minimum deletion.  It is not hard to show that the evolving
distribution of keys remaining in the heap is far from uniform.

\begin{lemma}
\label{lem:rand-ins-del}
After the initial $n$ insertions and $cn$ iterations of insert-delete, the items
remaining in the heap consist of the $n$ largest keys inserted thus far.  The
next item inserted has roughly a $\nicefrac{c}{c+1}$ probability of becoming the
new minimum.
\end{lemma}

\begin{proof}
For the purpose of this analysis, we consider the inserted keys to be
reals distributed uniformly at random in the range $\left[{0,1}\right]$, rather
than 32-bit integers. The pattern of operations leaves the
$n$ largest keys inserted thus far in the heap, as the following simple
inductive argument shows.  Initially, $n$ keys are inserted; being the
only keys thus far, they are trivially the largest.  Then, each iteration
consists of a single insertion followed by a minimum deletion.  Since there are
$n+1$ keys in the heap after the insertion, and the minimum is deleted, the
remaining $n$ keys are the largest thus far.

We can view the random variables of all keys inserted thus far to be
the collection $X_1, \dots, X_{\left({c+1}\right)n}$, and the current minimum in
the heap to be the $\left({cn + 1}\right)^{th}$ order statistic,
$X_{\left({cn+1}\right)}$.  The expectation of this variable is well-known:
$\ex{}{X_{\left({cn+1}\right)}} =
\nicefrac{cn+1}{\left({c+1}\right)n} \approx \nicefrac{c}{c+1}$.  From this we
deduce that the probability $p$ of the  next inserted key becoming the new
minimum is roughly $\nicefrac{c}{c+1}$.
\end{proof}

As $c$ grows, the most recent insertion becomes exceedingly likely to be the
next deleted item.  In other words, the behavior of the queue becomes
increasingly stack-like as the sequence lengthens.  On the other hand, if we introduce
$\DecreaseKeyOp$ operations to the sequence, we can ameliorate the degeneracy.
This brings us to our third type of artificial sequence.  We again build an
initial heap of size $n$ with random insertions.  We then perform $cn$
repetitions of the following: one random insertion, $k$ key decreases on random
nodes, and one minimum deletion.  We also consider two cases for the $k$
key decreases.  In the first, we decrease the key to some random number between
its current value and the minimum.  In the second, we decrease it so that it
becomes the new minimum.  We refer to these options as
``middle'' and ``min'', respectively.

In both the insertion-deletion workloads and the key-decrease workloads we consider $c \in \left\{{1,32,1024}\right\}$, while in the key-decrease workloads we also consider $k \in \left\{{1,32,1024}\right\}$.

\subsection{More realistic workloads.}

Of our remaining workloads, some are still artificial in the sense that they are
generated by running real algorithms on artificial inputs, but others make use
of real inputs.

The first two of these are Dijkstra's algorithm for single source shortest paths and the Nagamochi-Ibaraki algorithm for the min-cut problem.  We run both algorithms against well-structured or randomly generated graphs.  Dijkstra's algorithm in particular is run on several classes of graphs, including some which guarantee a $\DecreaseKeyOp$ operation for each edge.  Additionally, we run Dijkstra's algorithm on real road networks of different portions of the United States.

Our final set of trace files is generated from the \texttt{htsim} packet-level
network simulator \cite{htsim}, written by the authors of the multipath TCP
(MPTCP) protocol. The simulator models arbitrary networks using pipes (that add
delays) and queues (with fixed processing capacity and finite buffers), and
implements both TCP and MPTCP.  One of these workloads is based on real traffic
traces from the VL2 network~\cite{vl2}.

\section{Results}
\label{sec:results}

The results reveal a more nuanced truth than that which has been traditionally
accepted.  It is not true that implicit-4 heaps are optimal for all workloads,
nor is it true that Fibonacci heaps are always exceptionally slow.  We focus on
the most interesting cases here, and include the remaining results in the
\ifthenelse{\boolean{fullversion}}{appendix}{full version of our paper~\cite{fullversion}}.  We present most of our data in tables
sorted in ascending order of wallclock time.  Each table is for a single, large input file.  The
tables represent raw metrics divided by the minimum value attained by any heap,
such that a highlighted value of {\color{blue}\bf 1.00} is the minimum, while a
value $c$ is $c$ times the minimum.  These ratios make it
easier to interpret relative performance instead of the full counts.
In order to keep the tables compact, the column titles have been abbreviated:
\texttt{time} is wallclock time, \texttt{inst} is the dynamic instruction count,
\texttt{l1\_rd} and \texttt{l1\_wr} are the number of L1 reads and writes
respectively, \texttt{l2\_rd} and \texttt{l2\_wr} are the L2 reads and writes
respectively, \texttt{br} is the number of dynamic branches, and \texttt{l1\_m},
\texttt{l2\_m} and \texttt{br\_m} are the number of L1 misses, L2 misses, and
branch mispredictions.

\begin{figure*}
	\center
	\label{fig:log}
    \includegraphics[width=0.9\textwidth]{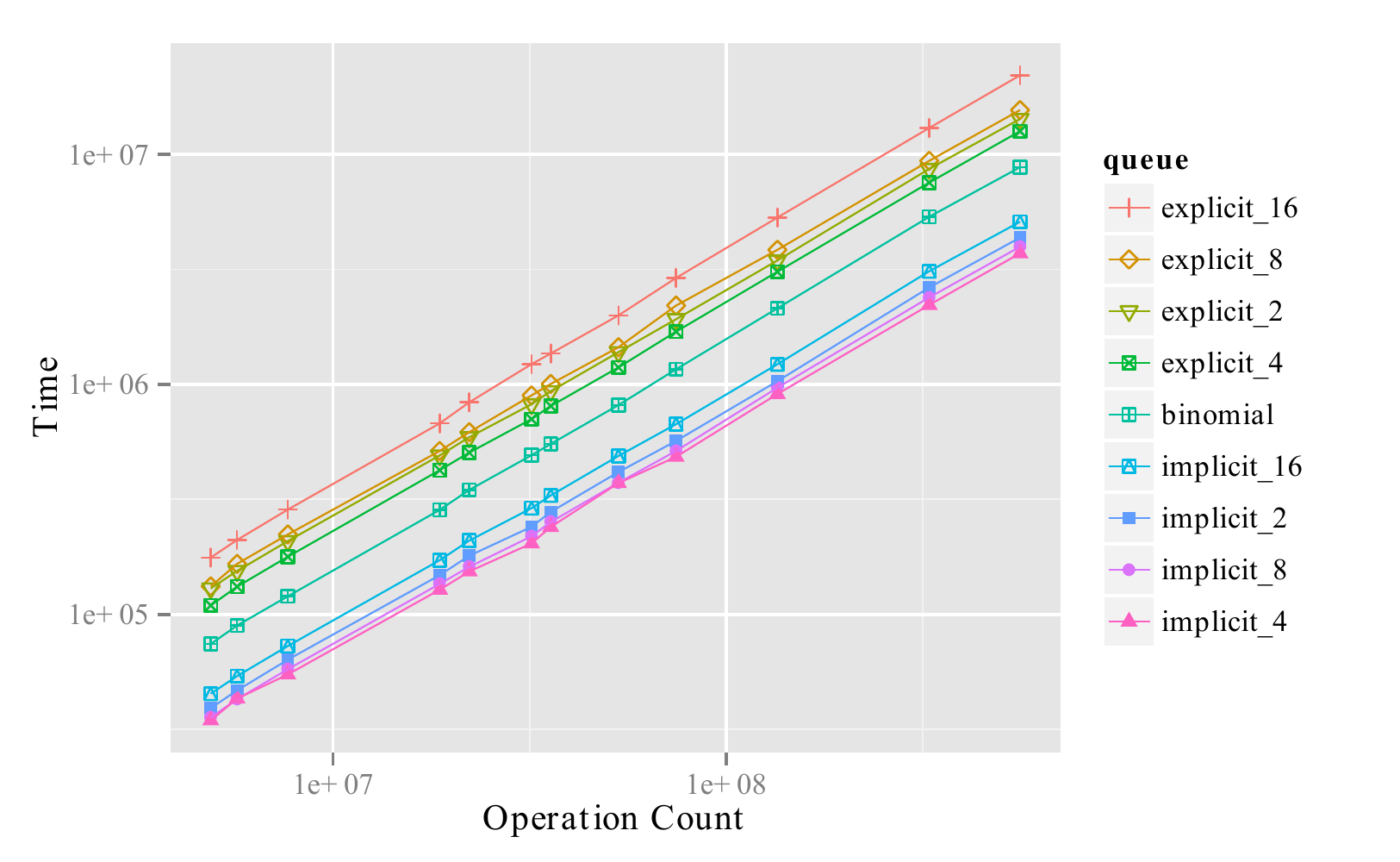}
	\caption{Dijkstra on the full USA road map. All operation counts are scaled by
	$\log n$.}
\end{figure*}

\begin{figure*}
	\center
	\label{fig:const}
    \includegraphics[width=0.9\textwidth]{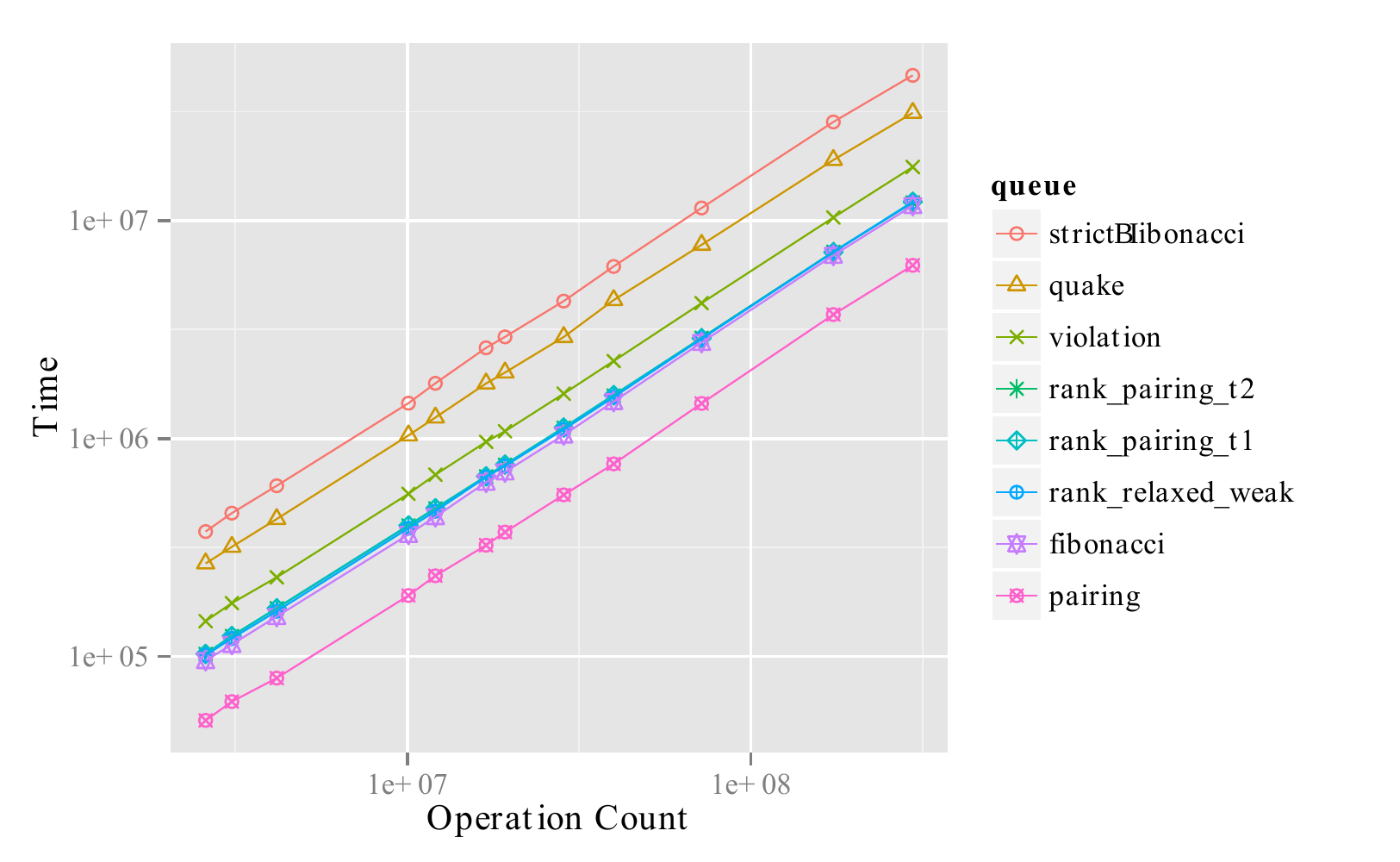}
	\caption{Dijkstra on the full USA road map. The $\DeleteMinOp$ count is scaled
	by $\log n$.}
\end{figure*}

We initially ran each experiment on many problem sizes.  We
found that in most cases the relative performance stabilized very quickly, so
from here on we only present data for the largest problem size.  See Figures 1
and 2 for some evidence of this stabilization.  The heaps are separated into two classes so as to unclutter the plots and give a consistent axis.
The operation counts are the sum of the counts of $\InsertOp$, $\DecreaseKeyOp$,
and $\DeleteMinOp$ operations.  In Figure 1, all operation counts are
scaled by $\log n$, where $n$ is the
average size of the heap.  In Figure 2, only the $\DeleteMinOp$ count is scaled
by $\log n$. This scaling approximately reflects the amortized bounds for each
heap.

Before diving into the results, we first make a high-level observation.
The number of L1 cache misses appears to be the metric most strongly correlated
with wallclock time.  It is not a perfect predictor, and inversions in ordering
certainly exist.  Some of these inversions can be explained by L2 cache misses,
write counts, or branch misprediction.  Others appear to be outliers or are
otherwise yet unexplained.

\begin{table*}
	\center
	\caption{Sorting}
	\label{tab:sorting}
		\begin{tabular}{r||cccccccccc}
		\multicolumn{11}{l}{Heap Size -- $\mathrm{max} = 4194304, \mathrm{average} = 2097152$} \\
		\multicolumn{11}{l}{Ratio of Operations -- $\textsc{Insert} : \textsc{DeleteMin} : \textsc{DecreaseKey} = 1.00 : 1.00 :       0.00$} \\ \hline
		queue & time & inst & l1\_rd & l1\_wr & l2\_rd & l2\_wr & br & l1\_m & l2\_m & br\_m \\ \hline
		implicit\_simple\_4 & {\color{blue}\bf 1.00} & {\color{blue}\bf 1.00} & {\color{blue}\bf 1.00} & 1.33 & {\color{blue}\bf 1.00} & 1.01 & {\color{blue}\bf 1.00} & {\color{blue}\bf 1.00} & {\color{blue}\bf 1.00} & {\color{blue}\bf 1.00} \\
		implicit\_simple\_8 & 1.05 & 1.23 & 1.09 & 1.11 & 1.01 & 1.00 & 1.12 & 1.01 & 1.01 & 1.12 \\
		implicit\_simple\_16 & 1.25 & 1.78 & 1.39 & {\color{blue}\bf 1.00} & 1.28 & {\color{blue}\bf 1.00} & 1.53 & 1.19 & 1.27 & 1.53 \\
		implicit\_simple\_2 & 1.59 & 1.12 & 1.23 & 1.98 & 1.01 & 1.04 & 1.29 & 1.31 & 1.01 & 1.29 \\
		implicit\_8 & 3.30 & 1.16 & 2.17 & 1.79 & 3.94 & 2.05 & 1.61 & 1.90 & 3.90 & 1.61 \\
\end{tabular}
\end{table*}

\begin{table*}
	\center
	\caption{Dijkstra -- full USA road map}
	\label{tab:usa}
		\begin{tabular}{r||cccccccccc}
		\multicolumn{11}{l}{Heap Size -- $\mathrm{max} = 4200, \mathrm{average} = 2489$} \\
		\multicolumn{11}{l}{Ratio of Operations -- $\textsc{Insert} : \textsc{DeleteMin} : \textsc{DecreaseKey} = 13.98 : 13.98 :  1.00$} \\ \hline
		queue & time & inst & l1\_rd & l1\_wr & l2\_rd & l2\_wr & br & l1\_m & l2\_m & br\_m \\ \hline
		implicit\_4 & {\color{blue}\bf 1.00} & {\color{blue}\bf 1.00} & {\color{blue}\bf 1.00} & 1.10 & 1.35 & 1.06 & {\color{blue}\bf 1.00} & {\color{blue}\bf 1.00} & 1.00 & {\color{blue}\bf 1.00} \\
		implicit\_8 & 1.07 & 1.12 & 1.12 & 1.03 & 1.61 & 1.18 & 1.01 & 1.07 & 1.20 & 1.01 \\
		implicit\_2 & 1.17 & 1.10 & 1.01 & 1.27 & 1.35 & {\color{blue}\bf 1.00} & 1.33 & 1.05 & {\color{blue}\bf 1.00} & 1.33 \\
		implicit\_16 & 1.37 & 1.42 & 1.38 & {\color{blue}\bf 1.00} & 2.20 & 1.35 & 1.21 & 1.24 & 1.63 & 1.21 \\
		pairing & 1.68 & 1.09 & 1.12 & 2.95 & 1.71 & 28.57 & 1.39 & 1.60 & 1.75 & 1.39 \\
\end{tabular}
\end{table*}

\subsection{\bf Conventional wisdom holds.}

We first examine two cases where the conventional wisdom holds.  As seen
in Table \ref{tab:sorting}, the implicit-simple heaps handle sorting workloads
very well.  The best performance is achieved by the implicit-simple-4 heap.  The
Fibonacci heap is almost seven times as slow as the fastest, which does indeed
echo old complaints about its speed.  The pairing heap and binomial queue fare
better here, but still poorly at at least four times as slow as the fastest.
Without any key decreases, the rank-relaxed weak queue is essentially just an
alternate implementation of a binomial queue, so it is not terribly surprising
that it does better than the Fibonacci heap.

Similarly with Dijkstra's algorithm on the full USA road map (Table
\ref{tab:usa}), we see implicit-4 heaps performing quite well, while Fibonacci
heaps are roughly three times as slow.  The explicit heaps are noticeably slower
even than Fibonacci heaps, and the only Fibonacci relaxation to perform well
here is the pairing heap.  The others are in fact {\em slower} than their
conceptual ancestor.  Although they exhibit similar caching behavior, their
code is somewhat more complicated, which may be contributing to the slowdown.

Both of the above workloads are very well-studied, and as such the
relative performance of the older heap variants should not be very surprising.

\begin{table*}
	\center
	\caption{Randomized $\InsertOp$--$\DeleteMinOp$ (Degenerate) -- $c=1024$}
	\label{tab:id_long}
		\begin{tabular}{r||cccccccccc}
		\multicolumn{11}{l}{Heap Size -- $\mathrm{max} = 131073, \mathrm{average} = 131041$} \\
		\multicolumn{11}{l}{Ratio of Operations -- $\textsc{Insert} : \textsc{DeleteMin} : \textsc{DecreaseKey} = 1.00 : 1.00 : 0.00$} \\ \hline
		queue & time & inst & l1\_rd & l1\_wr & l2\_rd & l2\_wr & br & l1\_m & l2\_m & br\_m \\ \hline
		pairing & {\color{blue}\bf 1.00} & {\color{blue}\bf 1.00} & {\color{blue}\bf 1.00} & 1.08 & 1.60 & 4.96 & 1.17 & {\color{blue}\bf 1.00} & 1.34 & 1.17 \\
		binomial & 1.14 & 1.02 & 1.08 & 1.12 & {\color{blue}\bf 1.00} & 5.44 & {\color{blue}\bf 1.00} & 1.06 & {\color{blue}\bf 1.00} & {\color{blue}\bf 1.00} \\
		rank\_relaxed\_weak & 1.48 & 1.43 & 1.50 & 1.49 & 1.62 & 1.09 & 1.91 & 1.45 & 1.07 & 1.91 \\
		fibonacci & 1.72 & 1.86 & 1.41 & 1.50 & 1.01 & 11.85 & 4.09 & 1.40 & 1.46 & 4.09 \\
		rank\_pairing\_t2 & 2.06 & 2.25 & 1.58 & 1.30 & 1.08 & 5.55 & 5.52 & 1.43 & 1.06 & 5.52 \\
\end{tabular}
\end{table*}

\subsection{\bf Degenerate results.}

We now turn to our randomized
insertion-deletion workload.  The results here are more surprising.  Recall
from Lemma~\ref{lem:rand-ins-del} that this workload is degenerate, in
that as the sequence goes on, the most recently inserted item is very
likely to be the next item deleted. Nevertheless, this sequence is commonly used
in empirical studies.
The shortest sequence we tested, $c=1$
\ifthenelse{\boolean{fullversion}}{(Table \ref{tab:id_short})}{},
remains rather close to the sorting workload.  On the other hand, when the
sequence is very long ($c=1024$), as shown in Table \ref{tab:id_long}, we
see a very different picture.  The queue-based structures outperform the
implicit heaps by a factor of at least two.  Under these assumptions about the
distribution, an $\InsertOp$ operation in a $d$-ary heap results in the
node being sifted all the way to the top, and the subsequent $\DeleteMinOp$ on
average results in another long sifting sequence.  In a queue structure with lazy insertion, the $\InsertOp$
commonly results in a singleton node which is simply removed afterwards with little to no restructuring.

Although degenerate in the above case, a generalization of this
sequence becomes a natural sequence for which efficient structures have been
designed. Consider workloads which frequently insert new
items near the minimum rather than toward the bottom of the heap.  Let
$\rank{x}$ denote the rank of $x$ among the items in the heap, such that the rank of the minimum is $1$ and the maximum is $n$.  Similarly let $\maxrank{x}$
be the
maximum value of $\rank{x}$ over the lifetime of $x$ in the heap.
Then there are structures which are optimized for both the case of frequently
deleting small-rank items and the case of frequently deleting large-rank items.
The fishspear data structure achieves an $\OO{\log \maxrank{x}}$ bound for
deletion, while rank-sensitive priority queues achieve an $\OO{\log
\left({\nicefrac{n}{\rank{x}}}\right)}$ bound
\cite{FischerSpear,DeanRankSensitive}.  Additonally, pairing heaps have been
shown to support $\DeleteMinOp$ in $\OO{\log k}$ time where $k$ is the number of
heap operations since the minimum item was inserted \cite{IaconoBound}.

The event simulation literature, largely orthogonal to the theory literature,
includes more sophisticated random models for generating insertion-deletion
workloads.  One in particular to note is the so-called ``classic hold'' model
which is essentially the same as the degenerate model, except that instead of
inserting a completely random key in each iteration, the new key is equal to the
most recently deleted key plus a positive random value.  This avoids the
degeneracy.  This and other models were explored in a previous experimental
study \cite{RonngrenStudy}.  That study also considers several special-case
priority queues with poor theoretical bounds (e.g., $\omega(\log n)$) which
nonetheless perform quite well for event simulation workloads.

\begin{table*}
	\center
	\caption{Randomized $\DecreaseKeyOp$ -- Middle, $c=1$, $k=1$}
	\label{tab:dcr_one_short}
		\begin{tabular}{r||cccccccccc}
		\multicolumn{11}{l}{Heap Size -- $\mathrm{max} = 8388609, \mathrm{average} = 7340032$} \\
		\multicolumn{11}{l}{Ratio of Operations -- $\textsc{Insert} : \textsc{DeleteMin} : \textsc{DecreaseKey} = 2.00 : 1.00 : 1.00$} \\ \hline
		queue & time & inst & l1\_rd & l1\_wr & l2\_rd & l2\_wr & br & l1\_m & l2\_m & br\_m \\ \hline
		implicit\_8 & {\color{blue}\bf 1.00} & 1.23 & 1.14 & 1.07 & 2.91 & 1.21 & {\color{blue}\bf 1.00} & 1.04 & 1.58 & {\color{blue}\bf 1.00} \\
		implicit\_4 & 1.05 & 1.15 & 1.03 & 1.19 & 2.57 & 1.63 & 1.07 & {\color{blue}\bf 1.00} & 1.43 & 1.07 \\
		binomial & 1.07 & 1.42 & 1.84 & 3.99 & 1.28 & 6.09 & 1.24 & 2.26 & 1.05 & 1.24 \\
		pairing & 1.09 & {\color{blue}\bf 1.00} & {\color{blue}\bf 1.00} & 2.90 & 1.52 & 4.92 & 1.04 & 1.41 & 1.10 & 1.04 \\
		implicit\_16 & 1.15 & 1.53 & 1.41 & {\color{blue}\bf 1.00} & 3.98 & 1.02 & 1.13 & 1.21 & 2.12 & 1.13 \\
\end{tabular}
\end{table*}

\begin{table*}
	\center
	\caption{Randomized $\DecreaseKeyOp$ -- Min, $c=1$, $k=1$}
	\label{tab:dcr_min_one_short}
		\begin{tabular}{r||cccccccccc}
		\multicolumn{11}{l}{Heap Size -- $\mathrm{max} = 8388609, \mathrm{average} = 7340032$} \\
		\multicolumn{11}{l}{Ratio of Operations -- $\textsc{Insert} : \textsc{DeleteMin} : \textsc{DecreaseKey} = 2.00 : 1.00 : 1.00$} \\ \hline
		queue & time & inst & l1\_rd & l1\_wr & l2\_rd & l2\_wr & br & l1\_m & l2\_m & br\_m \\ \hline
		pairing & {\color{blue}\bf 1.00} & {\color{blue}\bf 1.00} & {\color{blue}\bf 1.00} & {\color{blue}\bf 1.00} & {\color{blue}\bf 1.00} & 1.65 & {\color{blue}\bf 1.00} & {\color{blue}\bf 1.00} & {\color{blue}\bf 1.00} & {\color{blue}\bf 1.00} \\
		fibonacci & 3.04 & 2.62 & 2.21 & 2.58 & 3.00 & 3.52 & 4.91 & 2.37 & 2.43 & 4.91 \\
		rank\_relaxed\_weak & 3.65 & 3.16 & 3.05 & 2.18 & 4.66 & 2.01 & 5.89 & 2.68 & 2.40 & 5.89 \\
		rank\_pairing\_t2 & 5.43 & 3.42 & 2.57 & 1.96 & 5.94 & 1.58 & 7.54 & 2.32 & 2.67 & 7.54 \\
		rank\_pairing\_t1 & 5.95 & 3.39 & 2.56 & 1.95 & 5.94 & 1.57 & 7.46 & 2.30 & 2.67 & 7.46 \\
\end{tabular}
\end{table*}

\subsection{\bf Surpising results.}

As noted in our discussion of the workloads,
adding even a single key decrease per iteration to the random sequences lessens
the degeneracy.  Furthermore in Table \ref{tab:dcr_one_short} we see that if the
key decreases do not always generate a new minimum, as
would be the case in many applications such as graph search, then implicit heaps
with large branching factors continue to perform well.  When the key decreases always produce new minima, the amortized structures come out ahead, while worst-case structures
(implicit heaps and binomial queues included) fare poorly, as shown in Table
\ref{tab:dcr_min_one_short}.  As these sequences get longer, e.g. $c=1024$ and
$k=1$ \ifthenelse{\boolean{fullversion}}{(Table \ref{tab:dcr_min_one_long})}{}, the Fibonacci relaxations
gain ground, with rank-pairing-t1 heaps surpassing Fibonacci heaps.  We note
that the change in performance coincides with a large gap in L2 cache misses,
and is likely due to the long sifting process in $d$-ary heaps.

\begin{table*}
	\center
	\caption{Randomized $\DecreaseKeyOp$ -- Min, $c=1$, $k=1024$}
	\label{tab:dcr_min_many_short}
		\begin{tabular}{r||cccccccccc}
		\multicolumn{11}{l}{Heap Size -- $\mathrm{max} = 262145, \mathrm{average} = 262017$} \\
		\multicolumn{11}{l}{Ratio of Operations -- $\textsc{Insert} : \textsc{DeleteMin} : \textsc{DecreaseKey} =    2.00 :    1.00 : 1024.00$} \\ \hline
		queue & time & inst & l1\_rd & l1\_wr & l2\_rd & l2\_wr & br & l1\_m & l2\_m & br\_m \\ \hline
		explicit\_4 & {\color{blue}\bf 1.00} & 1.07 & 1.04 & 1.07 & 1.23 & 5.92 & 1.17 & 1.05 & 1.18 & 1.17 \\
		explicit\_16 & 1.07 & 1.11 & 1.08 & 1.11 & 1.50 & 6.22 & 1.30 & 1.09 & 1.36 & 1.30 \\
		implicit\_16 & 1.16 & 1.01 & 1.01 & {\color{blue}\bf 1.00} & 2.16 & 5.71 & 1.01 & 1.00 & 1.71 & 1.01 \\
		implicit\_2 & 1.22 & 1.01 & 1.00 & 1.01 & 2.09 & 5.71 & 1.03 & 1.00 & 1.67 & 1.03 \\
		binomial & 1.23 & 1.44 & 1.89 & 1.31 & 1.79 & 5.83 & 2.46 & 1.71 & 1.50 & 2.46 \\
\end{tabular}
\end{table*}

If we increase the density of key decreases in the sequence, then we see
something strange (Table \ref{tab:dcr_min_many_short}).  Suddenly the $d$-ary
heaps are doing well, and in particular the explicit heaps outperform the
implicit ones.  One possible explanation for this is that the level of
indirection in the implicit heap implementations requires them to not only
touch the same allocated nodes that the explicit heaps touch, but also to jump
around in the structural array while doing path traversals.  Noting that
implicit and explicit heaps have a similar number of L1 misses, this is one of
the few other workloads for which L2 behavior is a better performance predictor.

As to why the $d$-ary heaps outperform the amortized structures, consider the
overall pattern.  If many nodes have their keys decreased in a pairing heap, for
instance, then their subtrees are simply reattached underneath the root.
Each node access is likely to trigger a cache miss, as there will be little
revisiting of nodes other than the root, and the subsequent minimum deletion will
have to examine each of these nodes again in order to restructure the tree.
On the other hand, in the $d$-ary heaps, all the sifting is along ancestral
paths which share many nodes between operations, and hence the caching effects
are more favorable.

\subsection{\bf Other workloads.}

\ifthenelse{\boolean{fullversion}}{}{The full version of our paper contains a complete set of data and further
discussion of results~\cite{fullversion}.}  Among the other workloads we tested,
those that generate relatively small heap sizes
\ifthenelse{\boolean{fullversion}}{(e.g.\ Tables \ref{tab:grid_phard}, \ref{tab:grid_ssquare})}{}
favor the implicit heaps, while those operating on larger heap
sizes
\ifthenelse{\boolean{fullversion}}{(e.g.\ Tables \ref{tab:acyc_pos}, \ref{tab:grid_slong}, \ref{tab:grid_ssquare_s}, \ref{tab:rand_1_4}, \ref{tab:rand_4})}{}
favor amortized structures, especially
the pairing heap.  Very dense key-decrease workloads, including some ``bad''
inputs for Dijkstra's algorithm \ifthenelse{\boolean{fullversion}}{(Tables \ref{tab:spbad_dense} and
\ref{tab:spbad_sparse})}{} and the Nagamochi-Ibaraki workloads,
\ifthenelse{\boolean{fullversion}}{(Table \ref{tab:nix})}{}
favor implicit and explicit heaps.  The network
simulation workloads,
\ifthenelse{\boolean{fullversion}}{(Tables \ref{tab:htsim_ecmp_perm_fat}--\ref{tab:htsim_mptcp_vl2_over})}{}
which produce relatively small heap sizes and have no key decreases, all favor implicit-simple
heaps.

\section{Sanity Checks}
\label{sec:sanity}

We performed a few auxiliary experiments to verify our findings in the previous section.

\subsection{Testing the caching hypothesis.}

In order to lend some credence to our claim that caching is the primary predictor of performance in many of these test cases, we ran a few additional tests, tweaking the parameters of our implementations.  We added an extra padding field to the node in our pairing heap and implicit-4 heap implementations.  The extra field does not generate additional instructions in the code other than in the original memory allocation process (not included in the timing procedures) and thus the only change should be in the memory address allocated to the nodes.  This can affect both caching and branch prediction.  Through repeated doubling of node size, we find that even though the dynamic instruction count does not grow, the wall-clock time does---in fact, it grows roughly in proportion to the cache miss rate.  A less-pronounced effect also accompanies the growth of the misprediction rate.

\ignore{
\begin{table}
	\center
	\caption{Increasing node size to test caching effects.}
	\begin{tabular}{c|ccc|ccc}
		& \multicolumn{3}{c}{implicit} & \multicolumn{3}{c}{pairing} \\
		node size & time & rd & wr & time & rd & wr \\ \hline
		64 & 1,388,674 & 78,253,202 & 6,291,473 & 538,713 & 27,603,707 & 10,605,794 \\
		128 & 1,494,682 & 86,145,175 & 10,504,143 & 730,097 & 29,908,790 & 15,340,588 \\
		256 & 1,786,634 & 103,205,956 & 18,887,945 & 990,593 & 29,901,246 & 23,729,588 \\
	\end{tabular}
\end{table}
}

\begin{table}
	\center
	\caption{Tweaking node size to test caching effects.}
	\label{tab:cache}
	\vspace{5pt}
	\begin{tabular}{c|ccc|ccc}
		& \multicolumn{3}{c}{implicit} & \multicolumn{3}{c}{pairing} \\
		node size & time & rd & wr & time & rd & wr \\ \hline
		1.00 & 1.00 & 1.00 & 1.00 & 1.00 & 1.00 & 1.00 \\
		2.00 & 1.08 & 1.10 & 1.67 & 1.36 & 1.08 & 1.45 \\
		4.00 & 1.29 & 1.32 & 3.00 & 1.84 & 1.08 & 2.24 \\
	\end{tabular}
\end{table}

One potentially interesting observation from these experiments is
this: the instruction patterns of pairing heaps is \textit{write-first}, while
that of implicit heaps is \textit{read-first}.  By this we mean that typically,
whenever an implicit heap touches a node, it does so first via a read, while a
pairing heap quite often simply overwrites data in the node without reading it.
This means that the cache behavior for pairing heaps is skewed toward write
misses, while implicit heaps are skewed toward read misses.
Table~\ref{tab:cache} shows the read and write miss rates for both heaps.

\subsection{Comparison to an existing implementation.}

We ran a few
experiments against Sanders' implementation of the sequence
heap~\cite{SandersSequence}, which has a reputation of being
hard to beat in practice.  This gives us an easy way to benchmark our own
untuned implementations to see how they compare against a well-engineered one.
The results were encouraging.

Of the four workloads we tested, the sequence heap was faster than any of our
implementations on two of them, while it was slower on the other two.  More
specifically, the sequence heap was $1.97$ times faster than the
implicit-simple-4 heap on the sorting workload, and a significant $3.69$ times
faster than the pairing heap on the randomized insertion-deletion workload with
$c=32$.  Our pairing heap implementation performed $1.36$ times faster than the
sequence heap on the insertion-deletion workload with $c=1024$, and the
implicit-simple-2 heap was $1.15$ times faster on one of the network simulator
workloads.

\section{Remarks}
\label{sec:remarks}

As declared in the introduction, this is by no means a final study.  The push in
the past decade for better-performing Fibonacci-like heaps, while it may have
led to theoretical simplifications, does not seem to have yielded obvious
practical benefits. The results show that the optimal choice of implementation
is strongly input-dependent.  Furthermore, it shows that care must be taken to
optimize for cache performance, primarily at the L1-L2 barrier.  This suggests
that complicated, cache-oblivious structures are unlikely to perform well
compared to simpler, cache-aware structures.  Some obvious
candidates for renewed testing are sequence heaps and B-heaps~\cite{BHeap}.
Another obvious direction for future work is to explore other classes of workloads.

We hope that our study gives future theorists and practitioners a new
outlook on the state of affairs.  Unfortunately, there is no simple answer to
which heap should be used when.  Picking the best tool for the job will likely
require experimentation between existing implementations or careful analysis of the
expected workload's caching behavior against each heap. To this end, we hope
that our simple implementations of various heap structures will serve as a
useful resource.

\ignore{
Finally, we advocate teaching additional variants at the undergraduate
level, rather than letting the implicit heaps hog the spotlight. \comment{Do we have
evidence that this is the case? It's not true at Princeton\ldots} The pairing
heap quite often out-performs even the implicit-4 heap, and cache-aware
structures offer hope of doing even better.
}

\subsection*{Acknowledgments.} We would like to thank Jeff Erickson
for the considerable guidance he provided the first author in the initial stages of this
project.  Some of the compute cluster resources at
Princeton were donated by Yahoo!.

\balance

\setlength{\bibsep}{5pt}
\bibliographystyle{plain}
\bibliography{references}

\ignore{

}

\ifthenelse{\boolean{fullversion}}
{
\clearpage
\section*{Appendix}

Below we have included the data tables for the rest of our results.
Due to space constraints, we were not able to discuss all of them in detail, but
we offer a brief description of the workloads involved.

\paragraph{\bf Dijkstra on contrived graphs.}  In Tables
\ref{tab:acyc_pos}--\ref{tab:spbad_sparse}, we see the results of running
Dijkstra's shortest paths algorithm on graphs from the DIMACS generators.  In
general these are well-structured graphs, though some of them do contain
randomness. A more complete description of the generators is in
the DIMACS codebase \cite{DIMACS5,DIMACS9}.

\paragraph{\bf Nagamochi-Ibaraki.}  In Table \ref{tab:nix} we see the results of running the Nagamochi-Ibaraki algorithm for the minimum-cut problem on random graphs (again from the DIMACS generators).

\paragraph{\bf Further artificial workloads.}  In Tables
\ref{tab:id_short}--\ref{tab:dcr_min_many_long} we see the results of the rest
of our artificial workloads.  These include extra settings of the $c$ and $k$
parameters described in Section~\ref{sec:artificial}.

\paragraph{\bf Network event simulation.}  In tables
\ref{tab:htsim_ecmp_perm_fat}--\ref{tab:htsim_mptcp_vl2_over} we see the results
of running different protocols and workloads on the \texttt{htsim} network
simulator.  They vary between the use of the Equal-Cost Multi-Path
(ECMP)~\cite{ecmp} or MPTCP protocols, a permutation traffic matrix or
traffic from VL2 network traces~\cite{vl2}, and a 128-node fat-tree or a
512-node, 4-to-1 oversubscribed fat-tree network topology.  See
\cite{SenScalable} for further discussion of these inputs.

{\setlength{\tabcolsep}{0.1em}
\begin{table}[h]
	\center
	\caption{Dijkstra -- \texttt{acyc\_pos} graphs}
	\label{tab:acyc_pos}
		\begin{tabular}{r||cccccccccc}
		\multicolumn{11}{l}{Heap Size -- $\mathrm{max} = 3005438, \mathrm{average} = 1725537$} \\
		\multicolumn{11}{l}{Ratio of Operations -- $\textsc{Insert} : \textsc{DeleteMin} : \textsc{DecreaseKey} = 1.00 : 1.00 : 1.34$} \\ \hline
		queue & time & inst & l1\_rd & l1\_wr & l2\_rd & l2\_wr & br & l1\_m & l2\_m & br\_m \\ \hline
		pairing & {\color{blue}\bf 1.00} & {\color{blue}\bf 1.00} & 1.07 & 3.55 & {\color{blue}\bf 1.00} & 5.09 & 1.22 & 1.61 & {\color{blue}\bf 1.00} & 1.22 \\
		implicit\_4 & 1.03 & 1.02 & {\color{blue}\bf 1.00} & 1.20 & 1.85 & {\color{blue}\bf 1.00} & 1.05 & {\color{blue}\bf 1.00} & 1.31 & 1.05 \\
		implicit\_8 & 1.06 & 1.12 & 1.12 & 1.06 & 2.14 & 1.01 & {\color{blue}\bf 1.00} & 1.06 & 1.51 & {\color{blue}\bf 1.00} \\
		implicit\_16 & 1.25 & 1.46 & 1.44 & {\color{blue}\bf 1.00} & 2.95 & 1.02 & 1.19 & 1.27 & 2.06 & 1.19 \\
		implicit\_2 & 1.41 & 1.25 & 1.08 & 1.59 & 2.27 & 1.01 & 1.59 & 1.15 & 1.60 & 1.59 \\
\end{tabular}
\end{table}}

{\setlength{\tabcolsep}{0.1em}
\begin{table}
	\center
	\caption{Dijkstra -- \texttt{grid\_phard} graphs}
	\label{tab:grid_phard}
		\begin{tabular}{r||cccccccccc}
		\multicolumn{11}{l}{Heap Size -- $\mathrm{max} = 672, \mathrm{average} = 669$} \\
		\multicolumn{11}{l}{Ratio of Operations -- $\textsc{Insert} : \textsc{DeleteMin} : \textsc{DecreaseKey} = 1.00 : 1.00 : 5.00$} \\ \hline
		queue & time & inst & l1\_rd & l1\_wr & l2\_rd & l2\_wr & br & l1\_m & l2\_m & br\_m \\ \hline
		implicit\_8 & {\color{blue}\bf 1.00} & 1.02 & 1.04 & 1.03 & 2.16 & 9.41 & {\color{blue}\bf 1.00} & 1.00 & 2.29 & {\color{blue}\bf 1.00} \\
		implicit\_4 & 1.03 & {\color{blue}\bf 1.00} & {\color{blue}\bf 1.00} & 1.11 & 2.51 & 8.25 & 1.11 & {\color{blue}\bf 1.00} & 2.50 & 1.11 \\
		implicit\_16 & 1.15 & 1.26 & 1.22 & {\color{blue}\bf 1.00} & 2.34 & 9.61 & 1.16 & 1.12 & 2.44 & 1.16 \\
		implicit\_2 & 1.20 & 1.17 & 1.07 & 1.31 & 2.87 & 6.10 & 1.57 & 1.10 & 2.65 & 1.57 \\
		pairing & 1.44 & 1.15 & 1.17 & 2.75 & 1.02 & 3.37 & 1.55 & 1.59 & 1.02 & 1.55 \\
\end{tabular}
\end{table}}

{\setlength{\tabcolsep}{0.1em}
\begin{table}
	\center
	\caption{Dijkstra -- \texttt{grid\_slong} graphs}
	\label{tab:grid_slong}
		\begin{tabular}{r||cccccccccc}
		\multicolumn{11}{l}{Heap Size -- $\mathrm{max} = 4194305, \mathrm{average} = 2172571$} \\
		\multicolumn{11}{l}{Ratio of Operations -- $\textsc{Insert} : \textsc{DeleteMin} : \textsc{DecreaseKey} = 8.05 : 8.05 : 1.00$} \\ \hline
		queue & time & inst & l1\_rd & l1\_wr & l2\_rd & l2\_wr & br & l1\_m & l2\_m & br\_m \\ \hline
		pairing & {\color{blue}\bf 1.00} & {\color{blue}\bf 1.00} & {\color{blue}\bf 1.00} & 1.22 & 1.87 & 2.61 & {\color{blue}\bf 1.00} & {\color{blue}\bf 1.00} & 1.24 & {\color{blue}\bf 1.00} \\
		binomial & 1.20 & 1.89 & 1.70 & 1.66 & 1.93 & 2.93 & 1.91 & 1.55 & 1.31 & 1.91 \\
		fibonacci & 1.84 & 3.31 & 2.71 & 4.00 & {\color{blue}\bf 1.00} & 8.05 & 5.50 & 2.96 & 1.62 & 5.50 \\
		rank\_relaxed\_weak & 1.99 & 2.94 & 2.84 & 1.97 & 3.18 & 2.03 & 4.53 & 2.31 & 1.76 & 4.53 \\
		implicit\_4 & 2.03 & 2.78 & 1.96 & 1.19 & 7.28 & 1.19 & 3.50 & 1.54 & 3.53 & 3.50 \\
\end{tabular}
\end{table}}

{\setlength{\tabcolsep}{0.1em}
\begin{table}
	\center
	\caption{Dijkstra -- \texttt{grid\_ssquare} graphs}
	\label{tab:grid_ssquare}
		\begin{tabular}{r||cccccccccc}
		\multicolumn{11}{l}{Heap Size -- $\mathrm{max} = 11648, \mathrm{average} = 10915$} \\
		\multicolumn{11}{l}{Ratio of Operations -- $\textsc{Insert} : \textsc{DeleteMin} : \textsc{DecreaseKey} = 3.48 : 3.48 : 1.00$} \\ \hline
		queue & time & inst & l1\_rd & l1\_wr & l2\_rd & l2\_wr & br & l1\_m & l2\_m & br\_m \\ \hline
		implicit\_4 & {\color{blue}\bf 1.00} & {\color{blue}\bf 1.00} & {\color{blue}\bf 1.00} & 1.12 & 2.05 & 1.00 & {\color{blue}\bf 1.00} & {\color{blue}\bf 1.00} & 1.11 & {\color{blue}\bf 1.00} \\
		implicit\_8 & 1.03 & 1.12 & 1.14 & 1.04 & 2.49 & 1.05 & 1.01 & 1.08 & 1.34 & 1.01 \\
		implicit\_2 & 1.23 & 1.13 & 1.02 & 1.34 & 2.31 & {\color{blue}\bf 1.00} & 1.38 & 1.08 & 1.25 & 1.38 \\
		implicit\_16 & 1.41 & 1.46 & 1.44 & {\color{blue}\bf 1.00} & 3.56 & 1.14 & 1.22 & 1.28 & 1.91 & 1.22 \\
		pairing & 1.47 & 1.12 & 1.16 & 3.34 & 1.61 & 12.94 & 1.41 & 1.71 & 1.19 & 1.41 \\
\end{tabular}
\end{table}}

{\setlength{\tabcolsep}{0.1em}
\begin{table}
	\center
	\caption{Dijkstra -- \texttt{grid\_ssquare\_s} graphs}
	\label{tab:grid_ssquare_s}
		\begin{tabular}{r||cccccccccc}
		\multicolumn{11}{l}{Heap Size -- $\mathrm{max} = 4194305, \mathrm{average} = 2097721$} \\
		\multicolumn{11}{l}{Ratio of Operations -- $\textsc{Insert} : \textsc{DeleteMin} : \textsc{DecreaseKey} = 1.00 : 1.00 : 1.30$} \\ \hline
		queue & time & inst & l1\_rd & l1\_wr & l2\_rd & l2\_wr & br & l1\_m & l2\_m & br\_m \\ \hline
		pairing & {\color{blue}\bf 1.00} & {\color{blue}\bf 1.00} & {\color{blue}\bf 1.00} & 2.78 & {\color{blue}\bf 1.00} & 3.10 & {\color{blue}\bf 1.00} & 1.28 & {\color{blue}\bf 1.00} & {\color{blue}\bf 1.00} \\
		implicit\_8 & 2.04 & 1.44 & 1.28 & 1.07 & 3.68 & 1.01 & 1.13 & 1.06 & 2.44 & 1.13 \\
		fibonacci & 2.09 & 2.38 & 2.34 & 6.16 & 1.18 & 8.36 & 2.78 & 2.91 & 1.74 & 2.78 \\
		implicit\_4 & 2.17 & 1.32 & 1.14 & 1.21 & 3.54 & {\color{blue}\bf 1.00} & 1.19 & {\color{blue}\bf 1.00} & 2.35 & 1.19 \\
		implicit\_16 & 2.31 & 1.87 & 1.63 & {\color{blue}\bf 1.00} & 4.94 & 1.02 & 1.33 & 1.26 & 3.23 & 1.33 \\
\end{tabular}
\end{table}}

{\setlength{\tabcolsep}{0.1em}
\begin{table}
	\center
	\caption{Dijkstra -- \texttt{rand\_1\_4} graphs}
	\label{tab:rand_1_4}
		\begin{tabular}{r||cccccccccc}
		\multicolumn{11}{l}{Heap Size -- $\mathrm{max} = 250266, \mathrm{average} = 167912$} \\
		\multicolumn{11}{l}{Ratio of Operations -- $\textsc{Insert} : \textsc{DeleteMin} : \textsc{DecreaseKey} = 1.00 : 1.00 : 4.28$} \\ \hline
		queue & time & inst & l1\_rd & l1\_wr & l2\_rd & l2\_wr & br & l1\_m & l2\_m & br\_m \\ \hline
		implicit\_16 & {\color{blue}\bf 1.00} & 1.33 & 1.32 & {\color{blue}\bf 1.00} & 1.78 & 1.25 & 1.15 & 1.18 & 1.25 & 1.15 \\
		implicit\_8 & 1.01 & 1.06 & 1.09 & 1.05 & 1.53 & 1.24 & {\color{blue}\bf 1.00} & 1.03 & 1.09 & {\color{blue}\bf 1.00} \\
		implicit\_4 & 1.07 & {\color{blue}\bf 1.00} & {\color{blue}\bf 1.00} & 1.16 & 1.52 & 1.24 & 1.07 & {\color{blue}\bf 1.00} & 1.08 & 1.07 \\
		pairing & 1.10 & 1.29 & 1.30 & 3.98 & {\color{blue}\bf 1.00} & 3.61 & 1.73 & 1.92 & {\color{blue}\bf 1.00} & 1.73 \\
		implicit\_2 & 1.51 & 1.22 & 1.08 & 1.48 & 1.97 & 1.25 & 1.61 & 1.14 & 1.37 & 1.61 \\
\end{tabular}
\end{table}}

{\setlength{\tabcolsep}{0.1em}
\begin{table}
	\center
	\caption{Dijkstra -- \texttt{rand\_4} graphs}
	\label{tab:rand_4}
		\begin{tabular}{r||cccccccccc}
		\multicolumn{11}{l}{Heap Size -- $\mathrm{max} = 3560924, \mathrm{average} = 2333145$} \\
		\multicolumn{11}{l}{Ratio of Operations -- $\textsc{Insert} : \textsc{DeleteMin} : \textsc{DecreaseKey} = 1.01 : 1.01 : 1.00$} \\ \hline
		queue & time & inst & l1\_rd & l1\_wr & l2\_rd & l2\_wr & br & l1\_m & l2\_m & br\_m \\ \hline
		implicit\_8 & {\color{blue}\bf 1.00} & 1.16 & 1.13 & 1.07 & 1.97 & 1.01 & {\color{blue}\bf 1.00} & 1.06 & 1.42 & {\color{blue}\bf 1.00} \\
		pairing & 1.01 & {\color{blue}\bf 1.00} & 1.04 & 3.55 & {\color{blue}\bf 1.00} & 6.57 & 1.17 & 1.58 & {\color{blue}\bf 1.00} & 1.17 \\
		implicit\_4 & 1.03 & 1.06 & {\color{blue}\bf 1.00} & 1.20 & 1.71 & {\color{blue}\bf 1.00} & 1.05 & {\color{blue}\bf 1.00} & 1.23 & 1.05 \\
		implicit\_16 & 1.15 & 1.52 & 1.46 & {\color{blue}\bf 1.00} & 2.73 & 1.04 & 1.20 & 1.28 & 1.94 & 1.20 \\
		rank\_relaxed\_weak & 1.59 & 2.42 & 3.20 & 4.50 & 1.48 & 3.39 & 3.32 & 3.35 & 1.18 & 3.32 \\
\end{tabular}
\end{table}}

{\setlength{\tabcolsep}{0.1em}
\begin{table}
	\center
	\caption{Dijkstra -- \texttt{spbad\_dense} graphs}
	\label{tab:spbad_dense}
		\begin{tabular}{r||cccccccccc}
		\multicolumn{11}{l}{Heap Size -- $\mathrm{max} = 8192, \mathrm{average} = 5500$} \\
		\multicolumn{11}{l}{Ratio of Operations -- $\textsc{Insert} : \textsc{DeleteMin} : \textsc{DecreaseKey} =    1.00 :    1.00 : 4129.26$} \\ \hline
		queue & time & inst & l1\_rd & l1\_wr & l2\_rd & l2\_wr & br & l1\_m & l2\_m & br\_m \\ \hline
		explicit\_2 & {\color{blue}\bf 1.00} & 1.02 & 1.01 & 1.02 & 3.40 & 5.90 & 1.04 & 1.01 & 1.88 & 1.04 \\
		implicit\_8 & 1.06 & 1.00 & 1.00 & 1.00 & 4.61 & 1.90 & {\color{blue}\bf 1.00} & 1.00 & 1.19 & {\color{blue}\bf 1.00} \\
		implicit\_4 & 1.08 & {\color{blue}\bf 1.00} & {\color{blue}\bf 1.00} & 1.00 & 5.40 & 2.01 & 1.00 & {\color{blue}\bf 1.00} & 1.35 & 1.00 \\
		explicit\_4 & 1.12 & 1.01 & 1.01 & 1.01 & 2.87 & 8.28 & 1.03 & 1.01 & 2.32 & 1.03 \\
		implicit\_2 & 1.15 & 1.00 & 1.00 & 1.00 & 6.31 & 1.93 & 1.00 & 1.00 & 1.49 & 1.00 \\
\end{tabular}
\end{table}}

{\setlength{\tabcolsep}{0.1em}
\begin{table}
	\center
	\caption{Dijkstra -- \texttt{spbad\_sparse} graphs}
	\label{tab:spbad_sparse}
		\begin{tabular}{r||cccccccccc}
		\multicolumn{11}{l}{Heap Size -- $\mathrm{max} = 8, \mathrm{average} = 5$} \\
		\multicolumn{11}{l}{Ratio of Operations -- $\textsc{Insert} : \textsc{DeleteMin} : \textsc{DecreaseKey} = 1.00 : 1.00 : 3.69$} \\ \hline
		queue & time & inst & l1\_rd & l1\_wr & l2\_rd & l2\_wr & br & l1\_m & l2\_m & br\_m \\ \hline
		binomial & {\color{blue}\bf 1.00} & 1.03 & 1.20 & 1.21 & 1.00 & 1.21 & 1.17 & 1.20 & 1.02 & 1.17 \\
		implicit\_16 & 1.18 & 1.00 & {\color{blue}\bf 1.00} & {\color{blue}\bf 1.00} & 2.45 & 2.08 & {\color{blue}\bf 1.00} & {\color{blue}\bf 1.00} & 2.42 & {\color{blue}\bf 1.00} \\
		implicit\_8 & 1.19 & {\color{blue}\bf 1.00} & 1.00 & 1.00 & {\color{blue}\bf 1.00} & {\color{blue}\bf 1.00} & 1.00 & 1.00 & {\color{blue}\bf 1.00} & 1.00 \\
		implicit\_2 & 1.25 & 1.06 & 1.01 & 1.06 & 2.45 & 2.08 & 1.26 & 1.03 & 2.42 & 1.26 \\
		implicit\_4 & 1.26 & 1.05 & 1.02 & 1.03 & {\color{blue}\bf 1.00} & {\color{blue}\bf 1.00} & 1.18 & 1.02 & {\color{blue}\bf 1.00} & 1.18 \\
\end{tabular}
\end{table}}

{\setlength{\tabcolsep}{0.1em}
\begin{table}
	\center
	\caption{Nagamochi-Ibaraki}
	\label{tab:nix}
		\begin{tabular}{r||cccccccccc}
		\multicolumn{11}{l}{Heap Size -- $\mathrm{max} = 122605, \mathrm{average} = 61687$} \\
		\multicolumn{11}{l}{Ratio of Operations -- $\textsc{Insert} : \textsc{DeleteMin} : \textsc{DecreaseKey} =   1.00 :   1.00 : 244.14$} \\ \hline
		queue & time & inst & l1\_rd & l1\_wr & l2\_rd & l2\_wr & br & l1\_m & l2\_m & br\_m \\ \hline
		implicit\_16 & {\color{blue}\bf 1.00} & 1.02 & 1.02 & {\color{blue}\bf 1.00} & 1.57 & 2.68 & 1.01 & 1.01 & 1.33 & 1.01 \\
		implicit\_8 & 1.01 & {\color{blue}\bf 1.00} & 1.00 & 1.01 & 1.65 & 2.67 & {\color{blue}\bf 1.00} & {\color{blue}\bf 1.00} & 1.37 & {\color{blue}\bf 1.00} \\
		implicit\_2 & 1.33 & 1.05 & 1.02 & 1.05 & 1.82 & 2.67 & 1.18 & 1.03 & 1.48 & 1.18 \\
		explicit\_2 & 1.35 & 1.57 & 1.39 & 1.62 & {\color{blue}\bf 1.00} & 3.55 & 2.27 & 1.45 & 1.07 & 2.27 \\
		explicit\_8 & 1.35 & 1.40 & 1.30 & 1.50 & 1.28 & 3.94 & 2.04 & 1.36 & 1.30 & 2.04 \\
\end{tabular}
\end{table}}

{\setlength{\tabcolsep}{0.1em}
\begin{table}
	\center
	\caption{Randomized $\InsertOp$--$\DeleteMinOp$ (Broken) -- $c=1$}
	\label{tab:id_short}
		\begin{tabular}{r||cccccccccc}
		\multicolumn{11}{l}{Heap Size -- $\mathrm{max} = 4194305, \mathrm{average} = 3495253$} \\
		\multicolumn{11}{l}{Ratio of Operations -- $\textsc{Insert} : \textsc{DeleteMin} : \textsc{DecreaseKey} = 2.00 : 1.00 :       0.00$} \\ \hline
		queue & time & inst & l1\_rd & l1\_wr & l2\_rd & l2\_wr & br & l1\_m & l2\_m & br\_m \\ \hline
		implicit\_simple\_8 & {\color{blue}\bf 1.00} & 1.29 & 1.07 & 1.11 & 1.29 & 1.67 & 1.12 & {\color{blue}\bf 1.00} & {\color{blue}\bf 1.00} & 1.12 \\
		implicit\_simple\_4 & 1.00 & 1.06 & {\color{blue}\bf 1.00} & 1.31 & 1.33 & 3.02 & {\color{blue}\bf 1.00} & 1.00 & 1.10 & {\color{blue}\bf 1.00} \\
		implicit\_simple\_16 & 1.16 & 1.81 & 1.31 & {\color{blue}\bf 1.00} & 1.55 & {\color{blue}\bf 1.00} & 1.47 & 1.13 & 1.14 & 1.47 \\
		implicit\_simple\_2 & 1.67 & 1.23 & 1.28 & 1.95 & 1.47 & 7.10 & 1.35 & 1.36 & 1.42 & 1.35 \\
		binomial & 2.41 & 1.27 & 2.98 & 6.17 & 1.09 & 9.38 & 1.66 & 3.59 & 1.27 & 1.66 \\
\end{tabular}
\end{table}}

{\setlength{\tabcolsep}{0.1em}
\begin{table}
	\center
	\caption{Randomized $\InsertOp$--$\DeleteMinOp$ (Broken) -- $c=32$}
	\label{tab:id_medium}
		\begin{tabular}{r||cccccccccc}
		\multicolumn{11}{l}{Heap Size -- $\mathrm{max} = 4194305, \mathrm{average} = 4162041$} \\
		\multicolumn{11}{l}{Ratio of Operations -- $\textsc{Insert} : \textsc{DeleteMin} : \textsc{DecreaseKey} = 1.03 : 1.00 : 0.00$} \\ \hline
		queue & time & inst & l1\_rd & l1\_wr & l2\_rd & l2\_wr & br & l1\_m & l2\_m & br\_m \\ \hline
		binomial & {\color{blue}\bf 1.00} & 1.15 & 1.31 & 1.78 & {\color{blue}\bf 1.00} & 8.38 & {\color{blue}\bf 1.00} & 1.24 & {\color{blue}\bf 1.00} & {\color{blue}\bf 1.00} \\
		pairing & 1.16 & {\color{blue}\bf 1.00} & {\color{blue}\bf 1.00} & 1.54 & 1.62 & 7.66 & 1.03 & {\color{blue}\bf 1.00} & 1.35 & 1.03 \\
		rank\_relaxed\_weak & 1.20 & 1.64 & 1.82 & 2.09 & 1.61 & 1.59 & 1.88 & 1.61 & 1.07 & 1.88 \\
		implicit\_simple\_8 & 1.28 & 4.32 & 1.62 & 1.18 & 7.42 & 22.50 & 3.35 & 1.23 & 5.61 & 3.35 \\
		fibonacci & 1.42 & 1.98 & 1.60 & 2.25 & 1.00 & 18.94 & 2.99 & 1.53 & 1.48 & 2.99 \\
\end{tabular}
\end{table}}

{\setlength{\tabcolsep}{0.1em}
\begin{table}
	\center
	\caption{Randomized $\DecreaseKeyOp$ -- Middle, $c=32$, $k=1$}
	\label{tab:dcr_one_medium}
		\begin{tabular}{r||cccccccccc}
		\multicolumn{11}{l}{Heap Size -- $\mathrm{max} = 4194305, \mathrm{average} = 4172685$} \\
		\multicolumn{11}{l}{Ratio of Operations -- $\textsc{Insert} : \textsc{DeleteMin} : \textsc{DecreaseKey} = 1.03 : 1.00 : 1.00$} \\ \hline
		queue & time & inst & l1\_rd & l1\_wr & l2\_rd & l2\_wr & br & l1\_m & l2\_m & br\_m \\ \hline
		pairing & {\color{blue}\bf 1.00} & {\color{blue}\bf 1.00} & {\color{blue}\bf 1.00} & 1.45 & {\color{blue}\bf 1.00} & 2.82 & {\color{blue}\bf 1.00} & {\color{blue}\bf 1.00} & {\color{blue}\bf 1.00} & {\color{blue}\bf 1.00} \\
		explicit\_8 & 1.75 & 11.41 & 10.53 & 10.28 & 7.59 & 28.32 & 19.15 & 8.87 & 8.60 & 19.15 \\
		rank\_relaxed\_weak & 2.03 & 2.50 & 2.80 & 2.17 & 2.09 & 2.67 & 3.26 & 2.17 & 1.62 & 3.26 \\
		fibonacci & 2.04 & 2.62 & 2.44 & 3.54 & 1.40 & 5.56 & 3.51 & 2.44 & 1.64 & 3.51 \\
		rank\_pairing\_t1 & 2.69 & 2.89 & 2.53 & 2.04 & 1.28 & 2.23 & 4.36 & 1.99 & 1.08 & 4.36 \\
\end{tabular}
\end{table}}

{\setlength{\tabcolsep}{0.1em}
\begin{table}
	\center
	\caption{Randomized $\DecreaseKeyOp$ -- Middle, $c=1024$, $k=1$}
	\label{tab:dcr_one_long}
		\begin{tabular}{r||cccccccccc}
		\multicolumn{11}{l}{Heap Size -- $\mathrm{max} = 131073, \mathrm{average} = 131051$} \\
		\multicolumn{11}{l}{Ratio of Operations -- $\textsc{Insert} : \textsc{DeleteMin} : \textsc{DecreaseKey} = 1.00 : 1.00 : 1.00$} \\ \hline
		queue & time & inst & l1\_rd & l1\_wr & l2\_rd & l2\_wr & br & l1\_m & l2\_m & br\_m \\ \hline
		pairing & {\color{blue}\bf 1.00} & {\color{blue}\bf 1.00} & {\color{blue}\bf 1.00} & 1.45 & {\color{blue}\bf 1.00} & 2.57 & {\color{blue}\bf 1.00} & {\color{blue}\bf 1.00} & {\color{blue}\bf 1.00} & {\color{blue}\bf 1.00} \\
		implicit\_8 & 1.62 & 2.43 & 1.91 & 1.08 & 8.79 & 2.09 & 2.47 & 1.36 & 5.16 & 2.47 \\
		rank\_pairing\_t1 & 1.66 & 2.66 & 2.33 & 1.95 & 1.50 & 1.87 & 4.03 & 1.86 & 1.15 & 4.03 \\
		rank\_pairing\_t2 & 1.73 & 2.66 & 2.33 & 1.95 & 1.51 & 1.87 & 4.05 & 1.86 & 1.15 & 4.05 \\
		implicit\_4 & 1.73 & 2.37 & 1.78 & 1.25 & 8.50 & 3.24 & 2.80 & 1.34 & 5.21 & 2.80 \\
\end{tabular}
\end{table}}

{\setlength{\tabcolsep}{0.1em}
\begin{table}
	\center
	\caption{Randomized $\DecreaseKeyOp$ -- Middle, $c=1$, $k=32$}
	\label{tab:dcr_few_short}
		\begin{tabular}{r||cccccccccc}
		\multicolumn{11}{l}{Heap Size -- $\mathrm{max} = 4194305, \mathrm{average} = 4134386$} \\
		\multicolumn{11}{l}{Ratio of Operations -- $\textsc{Insert} : \textsc{DeleteMin} : \textsc{DecreaseKey} =  2.00 :  1.00 : 32.00$} \\ \hline
		queue & time & inst & l1\_rd & l1\_wr & l2\_rd & l2\_wr & br & l1\_m & l2\_m & br\_m \\ \hline
		explicit\_4 & {\color{blue}\bf 1.00} & 2.77 & 2.61 & 3.85 & 1.30 & 4.66 & 4.49 & 2.89 & 1.41 & 4.49 \\
		implicit\_8 & 1.04 & {\color{blue}\bf 1.00} & 1.02 & 1.04 & 1.93 & 1.30 & {\color{blue}\bf 1.00} & {\color{blue}\bf 1.00} & 1.22 & {\color{blue}\bf 1.00} \\
		implicit\_16 & 1.15 & 1.13 & 1.12 & {\color{blue}\bf 1.00} & 2.11 & 1.30 & 1.04 & 1.06 & 1.32 & 1.04 \\
		pairing & 1.15 & 1.91 & 1.70 & 4.69 & {\color{blue}\bf 1.00} & 3.03 & 2.84 & 2.49 & {\color{blue}\bf 1.00} & 2.84 \\
		explicit\_8 & 1.20 & 2.85 & 2.77 & 4.14 & 1.61 & 5.67 & 4.87 & 3.08 & 1.73 & 4.87 \\
\end{tabular}
\end{table}}

{\setlength{\tabcolsep}{0.1em}
\begin{table}
	\center
	\caption{Randomized $\DecreaseKeyOp$ -- Middle, $c=32$, $k=32$}
	\label{tab:dcr_few_medium}
		\begin{tabular}{r||cccccccccc}
		\multicolumn{11}{l}{Heap Size -- $\mathrm{max} = 262145, \mathrm{average} = 262025$} \\
		\multicolumn{11}{l}{Ratio of Operations -- $\textsc{Insert} : \textsc{DeleteMin} : \textsc{DecreaseKey} =  1.03 :  1.00 : 32.00$} \\ \hline
		queue & time & inst & l1\_rd & l1\_wr & l2\_rd & l2\_wr & br & l1\_m & l2\_m & br\_m \\ \hline
		explicit\_2 & {\color{blue}\bf 1.00} & 2.64 & 2.36 & 3.20 & {\color{blue}\bf 1.00} & 4.09 & 4.10 & 2.56 & {\color{blue}\bf 1.00} & 4.10 \\
		implicit\_16 & 1.05 & 1.22 & 1.16 & {\color{blue}\bf 1.00} & 2.26 & 2.09 & 1.13 & 1.10 & 1.43 & 1.13 \\
		explicit\_4 & 1.06 & 2.09 & 1.95 & 2.61 & 1.16 & 4.15 & 3.32 & 2.11 & 1.09 & 3.32 \\
		implicit\_8 & 1.14 & 1.04 & 1.04 & 1.02 & 2.03 & 2.09 & {\color{blue}\bf 1.00} & 1.02 & 1.31 & {\color{blue}\bf 1.00} \\
		pairing & 1.17 & 2.30 & 1.91 & 5.13 & 1.10 & 4.76 & 3.67 & 2.80 & 1.13 & 3.67 \\
\end{tabular}
\end{table}}

{\setlength{\tabcolsep}{0.1em}
\begin{table}
	\center
	\caption{Randomized $\DecreaseKeyOp$ -- Middle, $c=1024$, $k=32$}
	\label{tab:dcr_few_long}
		\begin{tabular}{r||cccccccccc}
		\multicolumn{11}{l}{Heap Size -- $\mathrm{max} = 8193, \mathrm{average} = 8193$} \\
		\multicolumn{11}{l}{Ratio of Operations -- $\textsc{Insert} : \textsc{DeleteMin} : \textsc{DecreaseKey} =  1.00 :  1.00 : 32.00$} \\ \hline
		queue & time & inst & l1\_rd & l1\_wr & l2\_rd & l2\_wr & br & l1\_m & l2\_m & br\_m \\ \hline
		implicit\_8 & {\color{blue}\bf 1.00} & 1.05 & 1.04 & 1.01 & 1.33 & 1.99 & {\color{blue}\bf 1.00} & 1.02 & 1.01 & {\color{blue}\bf 1.00} \\
		implicit\_16 & 1.03 & 1.19 & 1.13 & {\color{blue}\bf 1.00} & 1.29 & 2.07 & 1.10 & 1.08 & {\color{blue}\bf 1.00} & 1.10 \\
		implicit\_4 & 1.05 & {\color{blue}\bf 1.00} & {\color{blue}\bf 1.00} & 1.04 & 1.43 & 1.91 & 1.01 & {\color{blue}\bf 1.00} & 1.06 & 1.01 \\
		implicit\_2 & 1.13 & 1.07 & 1.02 & 1.12 & 1.63 & 1.83 & 1.25 & 1.03 & 1.17 & 1.25 \\
		explicit\_4 & 1.37 & 1.87 & 1.71 & 2.14 & 1.16 & 3.84 & 2.90 & 1.82 & 1.12 & 2.90 \\
\end{tabular}
\end{table}}

{\setlength{\tabcolsep}{0.1em}
\begin{table}
	\center
	\caption{Randomized $\DecreaseKeyOp$ -- Middle, $c=1$, $k=1024$}
	\label{tab:dcr_many_short}
		\begin{tabular}{r||cccccccccc}
		\multicolumn{11}{l}{Heap Size -- $\mathrm{max} = 262145, \mathrm{average} = 262017$} \\
		\multicolumn{11}{l}{Ratio of Operations -- $\textsc{Insert} : \textsc{DeleteMin} : \textsc{DecreaseKey} =    2.00 :    1.00 : 1024.00$} \\ \hline
		queue & time & inst & l1\_rd & l1\_wr & l2\_rd & l2\_wr & br & l1\_m & l2\_m & br\_m \\ \hline
		explicit\_4 & {\color{blue}\bf 1.00} & 1.13 & 1.09 & 1.14 & 1.15 & 2.56 & 1.34 & 1.10 & 1.09 & 1.34 \\
		explicit\_16 & 1.08 & 1.15 & 1.11 & 1.17 & 1.34 & 2.72 & 1.44 & 1.13 & 1.23 & 1.44 \\
		explicit\_2 & 1.22 & 1.22 & 1.14 & 1.22 & {\color{blue}\bf 1.00} & 2.55 & 1.52 & 1.16 & {\color{blue}\bf 1.00} & 1.52 \\
		implicit\_2 & 1.24 & 1.03 & 1.01 & 1.02 & 1.95 & 2.37 & 1.08 & 1.01 & 1.54 & 1.08 \\
		implicit\_16 & 1.25 & 1.01 & 1.00 & {\color{blue}\bf 1.00} & 1.89 & 2.37 & 1.00 & 1.00 & 1.50 & 1.00 \\
\end{tabular}
\end{table}}

{\setlength{\tabcolsep}{0.1em}
\begin{table}
	\center
	\caption{Randomized $\DecreaseKeyOp$ -- Middle, $c=32$, $k=1024$}
	\label{tab:dcr_many_medium}
		\begin{tabular}{r||cccccccccc}
		\multicolumn{11}{l}{Heap Size -- $\mathrm{max} = 8193, \mathrm{average} = 8193$} \\
		\multicolumn{11}{l}{Ratio of Operations -- $\textsc{Insert} : \textsc{DeleteMin} : \textsc{DecreaseKey} =    1.03 :    1.00 : 1024.00$} \\ \hline
		queue & time & inst & l1\_rd & l1\_wr & l2\_rd & l2\_wr & br & l1\_m & l2\_m & br\_m \\ \hline
		explicit\_2 & {\color{blue}\bf 1.00} & 1.07 & 1.04 & 1.06 & {\color{blue}\bf 1.00} & 5.35 & 1.16 & 1.05 & 1.03 & 1.16 \\
		explicit\_8 & 1.18 & 1.05 & 1.04 & 1.06 & 1.28 & 5.76 & 1.15 & 1.04 & 1.24 & 1.15 \\
		implicit\_16 & 1.18 & 1.01 & 1.01 & {\color{blue}\bf 1.00} & 1.09 & 4.38 & 1.01 & 1.00 & 1.02 & 1.01 \\
		explicit\_4 & 1.23 & 1.05 & 1.03 & 1.05 & 1.15 & 5.59 & 1.12 & 1.04 & 1.15 & 1.12 \\
		explicit\_16 & 1.24 & 1.08 & 1.06 & 1.08 & 1.19 & 5.89 & 1.23 & 1.07 & 1.19 & 1.23 \\
\end{tabular}
\end{table}}

{\setlength{\tabcolsep}{0.1em}
\begin{table}
	\center
	\caption{Randomized $\DecreaseKeyOp$ -- Middle, $c=1024$, $k=1024$}
	\label{tab:dcr_many_long}
		\begin{tabular}{r||cccccccccc}
		\multicolumn{11}{l}{Heap Size -- $\mathrm{max} = 257, \mathrm{average} = 257$} \\
		\multicolumn{11}{l}{Ratio of Operations -- $\textsc{Insert} : \textsc{DeleteMin} : \textsc{DecreaseKey} =    1.00 :    1.00 : 1024.00$} \\ \hline
		queue & time & inst & l1\_rd & l1\_wr & l2\_rd & l2\_wr & br & l1\_m & l2\_m & br\_m \\ \hline
		explicit\_16 & {\color{blue}\bf 1.00} & 1.05 & 1.04 & 1.18 & 4.48 & 33.83 & 1.15 & 1.04 & 5.44 & 1.15 \\
		implicit\_8 & 1.11 & 1.00 & 1.00 & 1.13 & 4.57 & 18.98 & 1.00 & 1.00 & 4.57 & 1.00 \\
		explicit\_8 & 1.50 & 1.04 & 1.03 & 1.17 & 1.68 & 8.14 & 1.10 & 1.03 & 1.75 & 1.10 \\
		explicit\_4 & 1.54 & 1.03 & 1.02 & 1.16 & 1.34 & 6.48 & 1.07 & 1.02 & 1.40 & 1.07 \\
		implicit\_16 & 2.37 & 1.01 & 1.00 & 1.12 & 4.23 & 20.31 & 1.00 & 1.00 & 4.40 & 1.00 \\
\end{tabular}
\end{table}}

{\setlength{\tabcolsep}{0.1em}
\begin{table}
	\center
	\caption{Randomized $\DecreaseKeyOp$ -- Min, $c=32$, $k=1$}
	\label{tab:dcr_min_one_medium}
		\begin{tabular}{r||cccccccccc}
		\multicolumn{11}{l}{Heap Size -- $\mathrm{max} = 4194305, \mathrm{average} = 4172685$} \\
		\multicolumn{11}{l}{Ratio of Operations -- $\textsc{Insert} : \textsc{DeleteMin} : \textsc{DecreaseKey} = 1.03 : 1.00 : 1.00$} \\ \hline
		queue & time & inst & l1\_rd & l1\_wr & l2\_rd & l2\_wr & br & l1\_m & l2\_m & br\_m \\ \hline
		pairing & {\color{blue}\bf 1.00} & {\color{blue}\bf 1.00} & {\color{blue}\bf 1.00} & {\color{blue}\bf 1.00} & {\color{blue}\bf 1.00} & 1.97 & {\color{blue}\bf 1.00} & {\color{blue}\bf 1.00} & {\color{blue}\bf 1.00} & {\color{blue}\bf 1.00} \\
		fibonacci & 2.19 & 2.97 & 2.51 & 2.86 & 1.56 & 3.61 & 5.17 & 2.66 & 1.67 & 5.17 \\
		rank\_relaxed\_weak & 2.51 & 3.72 & 3.57 & 2.41 & 2.24 & 2.28 & 6.51 & 3.08 & 1.78 & 6.51 \\
		rank\_pairing\_t1 & 3.19 & 4.08 & 2.99 & 2.17 & 1.39 & 1.38 & 8.36 & 2.64 & 1.10 & 8.36 \\
		rank\_pairing\_t2 & 3.26 & 4.09 & 2.99 & 2.17 & 1.40 & 1.39 & 8.38 & 2.64 & 1.11 & 8.38 \\
\end{tabular}
\end{table}}

{\setlength{\tabcolsep}{0.1em}
\begin{table}
	\center
	\caption{Randomized $\DecreaseKeyOp$ -- Min, $c=1024$, $k=1$}
	\label{tab:dcr_min_one_long}
		\begin{tabular}{r||cccccccccc}
		\multicolumn{11}{l}{Heap Size -- $\mathrm{max} = 131073, \mathrm{average} = 131051$} \\
		\multicolumn{11}{l}{Ratio of Operations -- $\textsc{Insert} : \textsc{DeleteMin} : \textsc{DecreaseKey} = 1.00 : 1.00 : 1.00$} \\ \hline
		queue & time & inst & l1\_rd & l1\_wr & l2\_rd & l2\_wr & br & l1\_m & l2\_m & br\_m \\ \hline
		pairing & {\color{blue}\bf 1.00} & {\color{blue}\bf 1.00} & {\color{blue}\bf 1.00} & {\color{blue}\bf 1.00} & {\color{blue}\bf 1.00} & 2.03 & {\color{blue}\bf 1.00} & {\color{blue}\bf 1.00} & {\color{blue}\bf 1.00} & {\color{blue}\bf 1.00} \\
		rank\_pairing\_t1 & 1.78 & 3.56 & 2.69 & 1.95 & 1.27 & 1.39 & 7.05 & 2.37 & 1.03 & 7.05 \\
		fibonacci & 1.81 & 2.73 & 2.35 & 2.68 & 1.47 & 3.64 & 4.54 & 2.49 & 1.60 & 4.54 \\
		rank\_pairing\_t2 & 1.91 & 3.57 & 2.69 & 1.95 & 1.28 & 1.39 & 7.07 & 2.37 & 1.04 & 7.07 \\
		rank\_relaxed\_weak & 2.20 & 3.57 & 3.48 & 2.41 & 2.12 & 2.31 & 6.11 & 3.02 & 1.72 & 6.11 \\
\end{tabular}
\end{table}}

{\setlength{\tabcolsep}{0.1em}
\begin{table}
	\center
	\caption{Randomized $\DecreaseKeyOp$ -- Min, $c=1$, $k=32$}
	\label{tab:dcr_min_few_short}
		\begin{tabular}{r||cccccccccc}
		\multicolumn{11}{l}{Heap Size -- $\mathrm{max} = 8388609, \mathrm{average} = 8268772$} \\
		\multicolumn{11}{l}{Ratio of Operations -- $\textsc{Insert} : \textsc{DeleteMin} : \textsc{DecreaseKey} =  2.00 :  1.00 : 32.00$} \\ \hline
		queue & time & inst & l1\_rd & l1\_wr & l2\_rd & l2\_wr & br & l1\_m & l2\_m & br\_m \\ \hline
		explicit\_2 & {\color{blue}\bf 1.00} & 2.81 & 2.46 & 3.27 & {\color{blue}\bf 1.00} & 2.81 & 4.30 & 2.65 & {\color{blue}\bf 1.00} & 4.30 \\
		explicit\_4 & 1.08 & 2.21 & 2.04 & 2.70 & 1.07 & 3.17 & 3.45 & 2.19 & 1.09 & 3.45 \\
		explicit\_8 & 1.18 & 2.31 & 2.21 & 3.02 & 1.35 & 4.02 & 3.86 & 2.40 & 1.38 & 3.86 \\
		binomial & 1.26 & 1.52 & 2.03 & 2.80 & 1.47 & 2.72 & 1.96 & 2.21 & 1.27 & 1.96 \\
		explicit\_16 & 1.31 & 2.80 & 2.73 & 3.73 & 1.52 & 5.33 & 5.13 & 2.96 & 1.67 & 5.13 \\
\end{tabular}
\end{table}}

{\setlength{\tabcolsep}{0.1em}
\begin{table}
	\center
	\caption{Randomized $\DecreaseKeyOp$ -- Min, $c=32$, $k=32$}
	\label{tab:dcr_min_few_medium}
		\begin{tabular}{r||cccccccccc}
		\multicolumn{11}{l}{Heap Size -- $\mathrm{max} = 262145, \mathrm{average} = 262025$} \\
		\multicolumn{11}{l}{Ratio of Operations -- $\textsc{Insert} : \textsc{DeleteMin} : \textsc{DecreaseKey} =  1.03 :  1.00 : 32.00$} \\ \hline
		queue & time & inst & l1\_rd & l1\_wr & l2\_rd & l2\_wr & br & l1\_m & l2\_m & br\_m \\ \hline
		explicit\_2 & {\color{blue}\bf 1.00} & 2.37 & 2.12 & 2.75 & {\color{blue}\bf 1.00} & 2.83 & 3.60 & 2.27 & {\color{blue}\bf 1.00} & 3.60 \\
		explicit\_8 & 1.05 & 2.09 & 2.00 & 2.64 & 1.44 & 3.05 & 3.50 & 2.15 & 1.28 & 3.50 \\
		explicit\_16 & 1.07 & 2.64 & 2.54 & 3.41 & 1.58 & 3.47 & 4.87 & 2.75 & 1.43 & 4.87 \\
		implicit\_16 & 1.16 & 1.24 & 1.18 & {\color{blue}\bf 1.00} & 2.13 & 2.36 & 1.14 & 1.11 & 1.57 & 1.14 \\
		binomial & 1.19 & 1.49 & 2.00 & 2.51 & 1.79 & 2.57 & 1.98 & 2.12 & 1.40 & 1.98 \\
\end{tabular}
\end{table}}

{\setlength{\tabcolsep}{0.1em}
\begin{table}
	\center
	\caption{Randomized $\DecreaseKeyOp$ -- Min, $c=1024$, $k=32$}
	\label{tab:dcr_min_few_long}
		\begin{tabular}{r||cccccccccc}
		\multicolumn{11}{l}{Heap Size -- $\mathrm{max} = 8193, \mathrm{average} = 8193$} \\
		\multicolumn{11}{l}{Ratio of Operations -- $\textsc{Insert} : \textsc{DeleteMin} : \textsc{DecreaseKey} =  1.00 :  1.00 : 32.00$} \\ \hline
		queue & time & inst & l1\_rd & l1\_wr & l2\_rd & l2\_wr & br & l1\_m & l2\_m & br\_m \\ \hline
		implicit\_8 & {\color{blue}\bf 1.00} & 1.05 & 1.04 & 1.01 & 1.30 & 2.13 & 1.00 & 1.02 & 1.04 & 1.00 \\
		implicit\_16 & 1.02 & 1.20 & 1.14 & {\color{blue}\bf 1.00} & 1.22 & 2.20 & 1.12 & 1.09 & {\color{blue}\bf 1.00} & 1.12 \\
		implicit\_4 & 1.03 & {\color{blue}\bf 1.00} & {\color{blue}\bf 1.00} & 1.03 & 1.43 & 2.05 & {\color{blue}\bf 1.00} & {\color{blue}\bf 1.00} & 1.11 & {\color{blue}\bf 1.00} \\
		implicit\_2 & 1.10 & 1.06 & 1.01 & 1.10 & 1.64 & 1.96 & 1.19 & 1.02 & 1.23 & 1.19 \\
		explicit\_4 & 1.32 & 1.80 & 1.65 & 2.03 & 1.19 & 3.69 & 2.76 & 1.75 & 1.15 & 2.76 \\
\end{tabular}
\end{table}}

{\setlength{\tabcolsep}{0.1em}
\begin{table}
	\center
	\caption{Randomized $\DecreaseKeyOp$ -- Min, $c=32$, $k=1024$}
	\label{tab:dcr_min_many_medium}
		\begin{tabular}{r||cccccccccc}
		\multicolumn{11}{l}{Heap Size -- $\mathrm{max} = 8193, \mathrm{average} = 8193$} \\
		\multicolumn{11}{l}{Ratio of Operations -- $\textsc{Insert} : \textsc{DeleteMin} : \textsc{DecreaseKey} =    1.03 :    1.00 : 1024.00$} \\ \hline
		queue & time & inst & l1\_rd & l1\_wr & l2\_rd & l2\_wr & br & l1\_m & l2\_m & br\_m \\ \hline
		explicit\_2 & {\color{blue}\bf 1.00} & 1.08 & 1.05 & 1.07 & 1.02 & 5.60 & 1.18 & 1.06 & 1.08 & 1.18 \\
		implicit\_16 & 1.02 & 1.01 & 1.01 & {\color{blue}\bf 1.00} & {\color{blue}\bf 1.00} & 4.56 & 1.01 & 1.00 & {\color{blue}\bf 1.00} & 1.01 \\
		implicit\_8 & 1.10 & 1.00 & 1.00 & 1.00 & 1.18 & 4.46 & {\color{blue}\bf 1.00} & 1.00 & 1.12 & {\color{blue}\bf 1.00} \\
		explicit\_16 & 1.10 & 1.09 & 1.07 & 1.10 & 1.20 & 6.15 & 1.27 & 1.08 & 1.24 & 1.27 \\
		implicit\_4 & 1.15 & {\color{blue}\bf 1.00} & {\color{blue}\bf 1.00} & 1.00 & 1.38 & 4.31 & 1.00 & {\color{blue}\bf 1.00} & 1.24 & 1.00 \\
\end{tabular}
\end{table}}

{\setlength{\tabcolsep}{0.1em}
\begin{table}
	\center
	\caption{Randomized $\DecreaseKeyOp$ -- Min, $c=1024$, $k=1024$}
	\label{tab:dcr_min_many_long}
		\begin{tabular}{r||cccccccccc}
		\multicolumn{11}{l}{Heap Size -- $\mathrm{max} = 257, \mathrm{average} = 257$} \\
		\multicolumn{11}{l}{Ratio of Operations -- $\textsc{Insert} : \textsc{DeleteMin} : \textsc{DecreaseKey} =    1.00 :    1.00 : 1024.00$} \\ \hline
		queue & time & inst & l1\_rd & l1\_wr & l2\_rd & l2\_wr & br & l1\_m & l2\_m & br\_m \\ \hline
		explicit\_4 & {\color{blue}\bf 1.00} & 1.12 & 1.08 & 1.28 & 1.34 & 6.56 & 1.31 & 1.10 & 1.39 & 1.31 \\
		explicit\_8 & 1.08 & 1.14 & 1.10 & 1.32 & 1.67 & 8.37 & 1.37 & 1.12 & 1.75 & 1.37 \\
		explicit\_2 & 1.15 & 1.17 & 1.11 & 1.33 & 1.10 & 4.47 & 1.40 & 1.13 & 1.09 & 1.40 \\
		explicit\_16 & 1.18 & 1.18 & 1.13 & 1.37 & 4.61 & 35.40 & 1.50 & 1.16 & 5.55 & 1.50 \\
		implicit\_4 & 1.58 & 1.00 & 1.00 & 1.13 & 5.24 & 16.42 & 1.02 & 1.00 & 4.88 & 1.02 \\
\end{tabular}
\end{table}}

{\setlength{\tabcolsep}{0.1em}
\begin{table}
	\center
	\caption{Network Simulation -- ECMP, PERM, fat-tree}
	\label{tab:htsim_ecmp_perm_fat}
		\begin{tabular}{r||cccccccccc}
		\multicolumn{11}{l}{Heap Size -- $\mathrm{max} = 1029, \mathrm{average} = 852$} \\
		\multicolumn{11}{l}{Ratio of Operations -- $\textsc{Insert} : \textsc{DeleteMin} : \textsc{DecreaseKey} = 1.00 : 1.00 :        0.00$} \\ \hline
		queue & time & inst & l1\_rd & l1\_wr & l2\_rd & l2\_wr & br & l1\_m & l2\_m & br\_m \\ \hline
		implicit\_simple\_2 & {\color{blue}\bf 1.00} & 1.07 & 1.12 & 1.48 & 1.32 & {\color{blue}\bf 1.00} & 1.20 & 1.17 & 1.28 & 1.20 \\
		implicit\_simple\_4 & 1.02 & 1.03 & {\color{blue}\bf 1.00} & 1.16 & 1.25 & 1.01 & {\color{blue}\bf 1.00} & {\color{blue}\bf 1.00} & 1.22 & {\color{blue}\bf 1.00} \\
		implicit\_simple\_8 & 1.30 & 1.27 & 1.09 & 1.05 & 1.23 & 1.04 & 1.15 & 1.02 & 1.21 & 1.15 \\
		implicit\_4 & 1.47 & 1.22 & 2.03 & 1.99 & 2.54 & 3.92 & 1.74 & 1.92 & 2.62 & 1.74 \\
		implicit\_8 & 1.55 & 1.36 & 2.29 & 1.88 & 2.32 & 3.63 & 1.75 & 2.05 & 2.40 & 1.75 \\
\end{tabular}
\end{table}}

{\setlength{\tabcolsep}{0.1em}
\begin{table}
	\center
	\caption{Network Simulation -- ECMP, PERM, over fat-tree}
	\label{tab:htsim_ecmp_perm_over}
		\begin{tabular}{r||cccccccccc}
		\multicolumn{11}{l}{Heap Size -- $\mathrm{max} = 1558, \mathrm{average} = 1376$} \\
		\multicolumn{11}{l}{Ratio of Operations -- $\textsc{Insert} : \textsc{DeleteMin} : \textsc{DecreaseKey} = 1.00 : 1.00 :        0.00$} \\ \hline
		queue & time & inst & l1\_rd & l1\_wr & l2\_rd & l2\_wr & br & l1\_m & l2\_m & br\_m \\ \hline
		implicit\_simple\_4 & {\color{blue}\bf 1.00} & {\color{blue}\bf 1.00} & {\color{blue}\bf 1.00} & 1.17 & 1.32 & 1.00 & {\color{blue}\bf 1.00} & {\color{blue}\bf 1.00} & 1.29 & {\color{blue}\bf 1.00} \\
		implicit\_simple\_2 & 1.01 & 1.07 & 1.14 & 1.51 & 1.45 & {\color{blue}\bf 1.00} & 1.24 & 1.20 & 1.42 & 1.24 \\
		implicit\_simple\_8 & 1.33 & 1.29 & 1.12 & 1.06 & 1.29 & 1.02 & 1.20 & 1.04 & 1.26 & 1.20 \\
		implicit\_4 & 1.42 & 1.17 & 2.02 & 1.98 & 15.50 & 4.10 & 1.72 & 1.91 & 14.69 & 1.72 \\
		implicit\_8 & 1.54 & 1.35 & 2.33 & 1.88 & 15.39 & 4.30 & 1.80 & 2.08 & 14.61 & 1.80 \\
\end{tabular}
\end{table}}

{\setlength{\tabcolsep}{0.1em}
\begin{table}
	\center
	\caption{Network Simulation -- ECMP, VL2, over fat-tree}
	\label{tab:htsim_ecmp_vl2_over}
		\begin{tabular}{r||cccccccccc}
		\multicolumn{11}{l}{Heap Size -- $\mathrm{max} = 10568, \mathrm{average} = 1944$} \\
		\multicolumn{11}{l}{Ratio of Operations -- $\textsc{Insert} : \textsc{DeleteMin} : \textsc{DecreaseKey} = 1.00 : 1.00 : 0.00$} \\ \hline
		queue & time & inst & l1\_rd & l1\_wr & l2\_rd & l2\_wr & br & l1\_m & l2\_m & br\_m \\ \hline
		implicit\_simple\_2 & {\color{blue}\bf 1.00} & 1.11 & 1.13 & 1.54 & 1.29 & 1.00 & 1.20 & 1.18 & 1.27 & 1.20 \\
		implicit\_simple\_4 & 1.01 & 1.06 & {\color{blue}\bf 1.00} & 1.19 & 1.20 & {\color{blue}\bf 1.00} & {\color{blue}\bf 1.00} & {\color{blue}\bf 1.00} & 1.19 & {\color{blue}\bf 1.00} \\
		implicit\_simple\_8 & 1.30 & 1.34 & 1.10 & 1.07 & 1.20 & 1.01 & 1.16 & 1.03 & 1.20 & 1.16 \\
		pairing & 1.77 & {\color{blue}\bf 1.00} & 1.78 & 3.94 & 27.66 & 309.07 & 1.58 & 2.32 & 37.01 & 1.58 \\
		implicit\_4 & 1.81 & 1.23 & 2.00 & 2.00 & 86.58 & 22.30 & 1.71 & 1.89 & 84.30 & 1.71 \\
\end{tabular}
\end{table}}

{\setlength{\tabcolsep}{0.1em}
\begin{table}
	\center
	\caption{Network Simulation -- MPTCP, PERM, fat-tree}
	\label{tab:htsim_mptcp_perm_fat}
		\begin{tabular}{r||cccccccccc}
		\multicolumn{11}{l}{Heap Size -- $\mathrm{max} = 1552, \mathrm{average} = 1417$} \\
		\multicolumn{11}{l}{Ratio of Operations -- $\textsc{Insert} : \textsc{DeleteMin} : \textsc{DecreaseKey} = 1.00 : 1.00 :        0.00$} \\ \hline
		queue & time & inst & l1\_rd & l1\_wr & l2\_rd & l2\_wr & br & l1\_m & l2\_m & br\_m \\ \hline
		implicit\_simple\_4 & {\color{blue}\bf 1.00} & 1.12 & {\color{blue}\bf 1.00} & 1.17 & 1.28 & {\color{blue}\bf 1.00} & {\color{blue}\bf 1.00} & {\color{blue}\bf 1.00} & 1.26 & {\color{blue}\bf 1.00} \\
		implicit\_simple\_2 & 1.00 & 1.21 & 1.14 & 1.52 & 1.39 & 1.02 & 1.24 & 1.20 & 1.37 & 1.24 \\
		implicit\_simple\_8 & 1.32 & 1.44 & 1.12 & 1.06 & 1.27 & 1.03 & 1.19 & 1.04 & 1.26 & 1.19 \\
		implicit\_4 & 1.39 & 1.32 & 2.02 & 1.99 & 10.92 & 5.14 & 1.73 & 1.90 & 10.62 & 1.73 \\
		pairing & 1.40 & {\color{blue}\bf 1.00} & 1.69 & 3.49 & 18.00 & 265.93 & 1.43 & 2.15 & 29.12 & 1.43 \\
\end{tabular}
\end{table}}

{\setlength{\tabcolsep}{0.1em}
\begin{table}
	\center
	\caption{Network Simulation -- MPTCP, PERM, over fat-tree}
	\label{tab:htsim_mptcp_perm_over}
		\begin{tabular}{r||cccccccccc}
		\multicolumn{11}{l}{Heap Size -- $\mathrm{max} = 4050, \mathrm{average} = 1624$} \\
		\multicolumn{11}{l}{Ratio of Operations -- $\textsc{Insert} : \textsc{DeleteMin} : \textsc{DecreaseKey} = 1.00 : 1.00 :        0.00$} \\ \hline
		queue & time & inst & l1\_rd & l1\_wr & l2\_rd & l2\_wr & br & l1\_m & l2\_m & br\_m \\ \hline
		implicit\_simple\_4 & {\color{blue}\bf 1.00} & 1.02 & {\color{blue}\bf 1.00} & 1.18 & 1.30 & 1.00 & {\color{blue}\bf 1.00} & {\color{blue}\bf 1.00} & 1.28 & {\color{blue}\bf 1.00} \\
		implicit\_simple\_2 & 1.00 & 1.08 & 1.14 & 1.53 & 1.41 & {\color{blue}\bf 1.00} & 1.22 & 1.19 & 1.39 & 1.22 \\
		implicit\_simple\_8 & 1.29 & 1.29 & 1.11 & 1.06 & 1.28 & 1.02 & 1.18 & 1.04 & 1.27 & 1.18 \\
		implicit\_4 & 1.55 & 1.19 & 2.01 & 1.99 & 62.06 & 13.10 & 1.72 & 1.90 & 59.88 & 1.72 \\
		implicit\_8 & 1.64 & 1.35 & 2.31 & 1.88 & 75.80 & 16.45 & 1.77 & 2.06 & 73.15 & 1.77 \\
\end{tabular}
\end{table}}

{\setlength{\tabcolsep}{0.1em}
\begin{table}
	\center
	\caption{Network Simulation -- MPTCP, VL2, over fat-tree}
	\label{tab:htsim_mptcp_vl2_over}
		\begin{tabular}{r||cccccccccc}
		\multicolumn{11}{l}{Heap Size -- $\mathrm{max} = 78229, \mathrm{average} = 2519$} \\
		\multicolumn{11}{l}{Ratio of Operations -- $\textsc{Insert} : \textsc{DeleteMin} : \textsc{DecreaseKey} = 1.00 : 1.00 :        0.00$} \\ \hline
		queue & time & inst & l1\_rd & l1\_wr & l2\_rd & l2\_wr & br & l1\_m & l2\_m & br\_m \\ \hline
		implicit\_simple\_2 & {\color{blue}\bf 1.00} & 1.12 & 1.13 & 1.54 & {\color{blue}\bf 1.00} & 1.01 & 1.21 & 1.19 & {\color{blue}\bf 1.00} & 1.21 \\
		implicit\_simple\_4 & 1.00 & 1.07 & {\color{blue}\bf 1.00} & 1.19 & 1.04 & {\color{blue}\bf 1.00} & {\color{blue}\bf 1.00} & {\color{blue}\bf 1.00} & 1.04 & {\color{blue}\bf 1.00} \\
		implicit\_simple\_8 & 1.30 & 1.34 & 1.10 & 1.07 & 1.15 & 1.00 & 1.16 & 1.03 & 1.15 & 1.16 \\
		pairing & 1.73 & {\color{blue}\bf 1.00} & 1.77 & 3.93 & 16.76 & 136.64 & 1.57 & 2.31 & 22.13 & 1.57 \\
		implicit\_4 & 1.91 & 1.23 & 2.00 & 2.00 & 53.64 & 11.56 & 1.71 & 1.89 & 51.73 & 1.71 \\
\end{tabular}
\end{table}}
}
{ % else
}

\end{document}